\documentclass[notitlepage]{revtex4-1}
\usepackage{amsmath}
\usepackage{amssymb,enumerate}
\usepackage{amsfonts}
\usepackage{amsthm}
\usepackage{latexsym}
\usepackage{graphicx}
\usepackage{hyperref}

\usepackage{color}

\newcommand{\be}{\begin{equation}}

\newcommand{\ee}{\end{equation}}
\newcommand{\eenn}{\end{equation*}}
\newcommand{\benn}{\begin{equation*}}
\newcommand{\bea}{\begin{eqnarray}}
\newcommand{\eea}{\end{eqnarray}}
\newcommand{\bra}[1]{\left\langle #1\right|}
\newcommand{\ket}[1]{\left|#1\right\rangle}
\newcommand{\braket}[2]{\left\langle #1,#2\right\rangle}
\newcommand{\ketbra}[2]{|#1\rangle\langle #2|}
\newcommand{\pure}[1]{\ketbra{#1}{#1}}
\newcommand{\Tr}{\mathop{\mathrm{Tr}}}
\providecommand{\one}{\leavevmode\hbox{\small1\kern-3.8pt\normalsize1}}

\newcommand{\ep}[1]{\mathrm{e}^{#1}}
\newcommand{\etal}{\textit{et al.}}

\newcommand{\Z}{\mathbb{Z}}
\newcommand{\R}{\mathbb{R}}

\newcommand{\F}{\mathcal{F}_{\gamma',s}}
\newcommand{\Fz}{\mathcal{F}_{\gamma',0}}
\newcommand{\G}{\mathcal{G}_{\gamma',s}}
\newcommand{\C}{\mathcal{C}}

\newcommand{\X}{\tilde{X}}
\newcommand{\W}{\hat{W}}
\renewcommand{\H}{\mathcal{H}}
\makeatletter
\makeatletter
\newtheorem*{rep@thm}{\rep@title}
\newcommand{\newreptheorem}[2]{%
\newenvironment{rep#1}[1]{%
 \def\rep@title{#2 \ref{##1}}%
 \begin{rep@thm}}%
 {\end{rep@thm}}}
\newtheorem{thm}{Theorem}
\newreptheorem{thm}{Theorem}
\newtheorem{defn}{Definition}
\newtheorem{example}{Example}
\newtheorem{cor}{Corollary}
\newtheorem{lem}{Lemma}

\newtheorem{prop}{Proposition}

\makeatother

\begin{document}

\title[Frustration-free stability]{Stability of Frustration-Free Hamiltonians}

\author{Spyridon \surname{Michalakis}}
\affiliation{Institute for Quantum Information and Matter,
Caltech,
Pasadena, CA 91125}
\email{spiros@caltech.edu}

\author{Justyna P. \surname{Zwolak}}
\affiliation{Department of Physics,
Nicolaus Copernicus University,
Torun, Poland}
\affiliation{Department of Physics,
Oregon State University,
Corvallis, OR 97331}
\email{letypka@gmail.com}

 \date{\today}
 
\begin{abstract}
We prove stability of the spectral gap for gapped, frustration-free Hamiltonians under general, quasi-local perturbations. We present a necessary and sufficient condition for stability, which we call {\it Local Topological Quantum Order} and show that this condition implies an area law for the entanglement entropy of the groundstate subspace. This result extends previous work by Bravyi \etal \, on the stability of topological quantum order for Hamiltonians composed of commuting projections with a common zero-energy subspace. We conclude with a list of open problems relevant to spectral gaps and topological quantum order.
\end{abstract}
\maketitle

\section{Introduction}
Recent interest in topological quantum computation has focused the attention of the condensed matter and mathematical physics community on Hamiltonians whose low-energy sectors exhibit some form of topological order. In a seminal paper \cite{kitaev:2003}, Kitaev constructed a Hamiltonian out of commuting spin-interaction terms, exhibiting four-fold groundstate degeneracy on the torus. More importantly, the four groundstates (known as the \textit{toric code}) are locally indistinguishable, implying protection against logical errors when encoding the state of two logical qubits in the toric code subspace. 

Nevertheless, if under a local perturbation of the Hamiltonian the groundstate subspace were to split appreciably, or become mixed with the non-topological higher-energy sectors, then any guarantee of protection from errors would no longer be valid. At this point, another type of stability against errors is required, this time at the level of the spacing between the low-energy and high-energy sectors and in the form of a negligible splitting in the groundstate subspace. Progress in this direction was made by Klich \cite{klich:2010} using the method of cluster expansions, and independently by Bravyi \etal \cite{bravyi:2010a,bravyi:2010b}. In particular, the latter result showed that all Hamiltonians with \textit{commuting} spin-interaction terms have negligible splitting in the groundstate subspace and a spectral gap that is stable under local perturbations, as long as the groundstates satisfy some type of \textit{local indistinguishability} and \textit{frustration-freeness} condition. Since the toric code Hamiltonian satisfies both conditions for gap stability, it follows that encoded qubits in the toric code groundstate subspace are robust against local errors, even if the Hamiltonian interactions are not precisely engineered. Moreover, at low enough temperatures, the spectral gap implies that the creation of excitations is unlikely to happen in relatively small time scales.

Even so, the toric code has a seemingly fatal flaw: at non-zero temperature, local excitations, once created, can travel around the torus to create logical errors at no extra energy cost. Motivated by this issue, soon after the results of Klich and Bravyi \etal, a series of papers appeared that focused on the beneficial effects of Anderson-type localization of the undesired excitations in the presence of impurities, or external magnetic fields \cite{kay:2011,stark:2011,wootton:2011}. 

Still, until recently, no Hamiltonian model was known that combined the desirable properties of the toric code, with a rigorous and sufficiently large lower bound on the energy barrier restricting the mobility of unwanted excitations at non-zero temperature. In fact, inspired by the pioneering work of Nussinov and Ortiz~\cite{nussinov:2008}, Yoshida showed in \cite{yoshida:2011} that such Hamiltonians would need to either forgo commutativity or break certain ``natural'' conditions, such as translational invariance or scale invariance, in order to satisfy stability at non-zero temperature. Indeed, the only known family of Hamiltonians with all the desired properties, recently presented by Haah in \cite{haah:2011}, has logical operators with fractal geometry in three dimensions and no scale invariance, thus sidestepping Yoshida's no-go theorem. Yet, despite the rigorous bound \cite{bravyi:2011} on the energy barrier of Haah's Hamiltonian, there remains a question of whether the barrier is large enough to allow for operations on the logical qubits, like read-out and error-correction, which may require times comparable to the time it takes for logical errors to appear. Moreover, it was recently shown by Hastings \cite{hastings:2011}, that all two-dimensional Hamiltonians which are a sum of commuting terms, have no topological order at non-zero temperature.

Motivated by this line of research and the larger question of the classification of quantum phases \cite{sachdev:2000, chen:2010b, BMNS:2011, schuch:2010a}, we present here a generalization of the result by Bravyi \etal, which removes the commutativity of the Hamiltonian terms as an assumption for stability. Some of the new candidate Hamiltonians now include, parent Hamiltonians of Matrix Product States (MPS) and Projected Entangled Pair States (PEPS) \cite{fannes:1992, P-GVCW:2008,schuch:2010,schuch:2010a}, as well as all other \textit{frustration-free} Hamiltonians - that is, Hamiltonians whose groundstates minimize the energy of each local interaction term. Moreover, we generalize the conditions needed for the stability of the spectral gap, in hopes that in the future, one may be able to prove an equivalent result for general, gapped Hamiltonians, whose groundstates satisfy some type of topological order.

The stability of quantum phases was already studied by Borgs \etal \cite{borgs:1996} and Datta \etal \cite{datta:1996a,datta:1996b}, where ``classical'' systems (Hamiltonians diagonalizable in a product basis) were shown to be robust against small quantum perturbations up to some low-temperature, using the methods of contour and cluster expansions. Nevertheless, here, we draw heavily from the methods developed in \cite{bravyi:2010a,bravyi:2010b}, following the more succinct format of \cite{bravyi:2010b}, which uses Hastings' powerful \textit{quasi-adiabatic continuation} \cite{hastings:2005}. We make extensive use of Lieb-Robinson bounds \cite{lieb1972}, both for the evolution of operators according to the Hamiltonians under consideration \cite{hastings:2004a,hastings:2006, nachtergaele:2006a,nachtergaele:2006b,nachtergaele:2007,nachtergaele:2009,nachtergaele:2010,PS09,PS10}, as well as Lieb-Robinson bounds on the quasi-adiabatic evolution, or \textit{spectral flow}, of gapped eigenspaces \cite{hastings:2009, hastings:2010, BMNS:2011}.

The paper is organized as follows: In Sec.~\ref{sec:ham}, we define the class of Hamiltonians whose stability we proceed to study. Sec.~\ref{sec:proof_sketch} contains the statement of the main result and the overall structure of the proof. In Sec.~\ref{sec:assumptions}, we introduce the conditions sufficient for proving stability, clarifying with examples the extend to which the conditions are also necessary. In Sec.~\ref{sec:perturbation}, we transform the perturbed Hamiltonian through a unitary process and a global energy shift in a form amenable to studying its low energy sectors. In Sec.~\ref{sec:relative_bound} we prove a relative bound on the energy added by the perturbation with respect to the energy of the unperturbed Hamiltonian. The proof of the main theorem appears in Sec.~\ref{sec:theorem} and in Sec.~\ref{sec:discussion} we close with a discussion of future work and a list of open problems. We present the proofs of Corollaries ~\ref{cor:local_tqo} and~\ref{cor:local_same} in Appendices \ref{append:local_tqo} and \ref{append:local_same}, respectively. Finally, Appendix~\ref{lr-estimates2} contains the technical details behind the estimate in Lemma~\ref{lemma:local}.


\section{\label{sec:ham} The Hamiltonian}
We study the stability of the spectral gap for Hamiltonians $H_0$ defined on $\Lambda = [0,L]^d \subset \Z^d$, with the corresponding Hilbert space being a tensor product of the local Hilbert spaces, $\H=\bigotimes_{u\in\Lambda}\H_u$, satisfying:
\begin{enumerate}
\item  ({\bf Spatially-local}) The unperturbed Hamiltonian can be written as a sum of geometrically local projectors $H_0 = \sum_{u \in \Lambda} Q_{u}$, where each interaction $Q_u$ acts non-trivially on a Hilbert space supported on $b_u(1)$, with $b_u(r)$ denoting the ball of radius $r$ with center at site $u \in \Lambda$.
\item ({\bf Periodic-boundary}) The Hamiltonian $H_0$ satisfies periodic boundary conditions.
\item ({\bf Frustration-free}) For $P_0$ the projector onto the groundstate subspace of $H_0$, we have $H_0 P_0 = 0$.
\item ({\bf Gapped}) $H_0$ has a spectral gap $\gamma_L \ge \gamma > 0$ above the groundstate subspace, for all $L \ge 2$.
\end{enumerate}

Recall that $L$ is the linear size of the support of $H_0$, so when we write $H_0$ we assume a fixed size $L$. The gap condition implies that as $H_0$ is defined (translationally invariant, or otherwise) for larger $L$, the gap $\gamma_L$ remains uniformly bounded from below. The restriction of each interaction on a ball of radius $1$ is not crucial, due to coarse-graining, as long as the interactions have constant range of support. Moreover, the assumption that the interactions are projections is not crucial either, as long as $Q_u P_0 = \lambda_u P_0$, where $\lambda_u$ is the minimum eigenvalue of $Q_u$. This condition is the most general way of defining a frustration-free Hamiltonian. From now on, however, we will assume that $\lambda_u = 0$, which follows from substituting $Q_u$ with $Q_u-\lambda_u \one$. Notice that this transformation corresponds to a global energy shift by $-\sum_u \lambda_u \, \one$, which leaves the spectral gap, and all other relevant properties of the Hamiltonian, unchanged.

\section{The main result}\label{sec:proof_sketch}
We begin this section by defining the general class of perturbations we will be studying. The only restriction is a certain notion of quasi-locality that we make precise in the following definition:
\begin{defn}\label{defn:perturbation}
We define a perturbation $V$ to have strength $J$ and decay $f(r)$, with $f(r) \le 1, \, r\ge 0$, iff we can write  $V = \sum_{u \in \Lambda} \sum_{r \ge 0}^L V_u(r)$, such that $V_u(r)$ has support on $b_u(r)$ and satisfies $\|V_u(r)\| \le J f(r)$, for some rapidly decaying function $f(r)$. We write that $V$ is a $(J,f)$-perturbation.
\end{defn}
As we discuss below, the proof involves a series of transformations on the initial $(J,f)$-perturbation $V$ into perturbations $V'$ with decay-function $f'$, making $V'$ a $(J\|f'\|, f'/\|f'\|)$-perturbation, where $\|f'\| = \sup_{r\ge 0} f'(r)$. From now on, for the sake of notational brevity, we will denote a $(J\|f\|,f/\|f\|)$-perturbation simply as a $(J,f)$-perturbation.

Moreover, we have chosen to leave the precise decay rate of the decay-function $f$ ill-defined, beyond saying that it is \emph{rapidly-decaying}, in order to avoid unnecessarily constraining the results that follow. Of course, throughout this paper, we assume that $f_0$ decays fast enough to allow for the usual Lieb-Robinson bounds \cite{hastings:2006, nachtergaele:2010} on evolutions based on $H_0+V$. This implies that fixed range, exponentially-decaying and even fast-enough polynomially-decaying perturbations are within the purview of this paper. In particular, for perturbations on $\Z^d$, a polynomial decay faster than $f(r) = (1+r)^{-(d+2)}$ suffices.

Our main result is the following:
\begin{thm}\label{thm:main_result}
Let $H_0$ be a {\it frustration-free} Hamiltonian with spectral gap $\gamma$, defined in Section \ref{sec:ham}, satisfying the {\it Local-TQO} (Definition \ref{def:local_tqo}) and {\it Local-Gap} (Definition \ref{def:local_gap}) conditions. For a $(J,f_0)$-perturbation $V$, there exist constants $J_0 >0$ and $L_0 \ge 2$, given in Proposition \ref{prop:relative_bound}, such that for $J \le J_0$ and $L \ge L_0$, the spectral gap of $H_0+V$ is bounded from below by $\gamma/2$.
\end{thm}
Let us sketch the overall structure of the proof. The complete proof is presented in Section \ref{sec:theorem}.

In order to relate the spectrum of the initial Hamiltonian $H_0$ to the spectrum of the final Hamiltonian $H_0+V$, we define the following one-parameter family of gapped Hamiltonians:
\begin{defn}\label{defn:path}
Let $H_s := H_0 + s V$, where $H_0$ satisfies the assumptions of the previous sections and $V$ is a $(J,f_0)$-perturbation. We define $s_0 \in [0,1)$ to be the maximum value of $s$, such that $H_s$ has spectral gap $\gamma_s \ge \gamma' > 0$ and groundstate projection $P_0(s)$, for $s \in [0,s_0]$.
\end{defn}
Note that $s_0$ may depend on $\gamma', \, J$ and $f_0$, as well as the linear size $L$, of the system. Moreover, $\gamma'$ is fixed a priori to some value less than $\gamma$. Concretely, we will take $\gamma' = \gamma/2$, but this is not crucial. The estimates below would, otherwise, depend implicitly on our choice of $\gamma'$.

The next steps are designed with one goal in mind: To show that there exist $J_0 > 0$ and $L_0 \ge 2$, such that for all $J \le J_0$ and $L\ge L_0$, we have that the spectral gap of $H_{s_0}$ is greater than $\frac{3}{4}\gamma$. This leads to a contradiction which is traced back to the definition of $s_0$ and our initial assumption that $s_0 < 1$.

\begin{enumerate}
\item First, we show that for $s\in[0,s_0]$, there exists a unitary $U(s)$ transforming the family of gapped Hamiltonians $H_s$ into $U^\dagger(s) H_s U(s) = H_0 +\sum_u X_u$, where the quasi-local terms $X_u$ satisfy $[X_u,P_0]=0$ and $U(s) P_0(0) U^\dagger(s)=P_0(s)$.

\item In the next step, we decompose $X_u=W_u+\Delta_u+c(X_u)\one$, where $W = \sum_{u\in \Lambda} W_u$ is a $(J,w)$-perturbation satisfying $W_u(r)P_{b_u(r)}=0$, with $P_{b_u(r)}$ given in Definition \ref{defn:local_gs}. Moreover, we show that the norm of $\Delta_u$ decays rapidly in $L$ and $c(X_u)=\text{Tr}(X_uP_0)/\text{Tr}P_0$ (see Proposition \ref{prop:local_decomposition} and Lemma \ref{lemma:D_u}).

\item Combining the {\it Local-Gap} condition with properties of the perturbation $W$, we prove that there exists a constant $c$ such that for an arbitrary state $\ket{\psi}$: $|\bra{\psi} W \ket{\psi}| \le c\cdot J\, \bra{\psi} H_0 \ket{\psi}$ (see Proposition \ref{prop:relative_bound}).

\item Finally, we use the relative bound from the previous step to show that there exist $J_0 > 0$ and $L_0 \ge 2$, such that for $L\ge L_0$ and $J\le J_0$, the Hamiltonian $H_0+sV$ has a gap $\gamma_s\ge (3/4)\gamma$, for all $s\in[0,s_0]$, thus obtaining a contradiction.
\end{enumerate}


\section{Assumptions for Stability}\label{sec:assumptions}
We begin this section by defining the following projections onto low-energy subspaces:
\begin{defn}\label{defn:local_gs}
For $B = b_u(r), \, r \ge 2$, we define $P_B(\epsilon)$ to be the projection onto the subspace of eigenstates of $H_B :=\sum_{b_v(1)\subset B} Q_v$ with energy at most $\epsilon \ge 0$. We set $P_B := P_B(0)$.
\end{defn}
The above definition implies the following relations, which hold for all frustration-free systems:
\begin{cor}\label{cor:local_energy} The following claims follow from the definition of $P_B(\epsilon)$:
\begin{itemize}
\item $H_B \ge \epsilon \,(1-P_B(\epsilon))$.
\item $P_B(\epsilon) P_B(\epsilon') = P_B(\epsilon')$, for $0 \le \epsilon' \le \epsilon$ and $P_{\Lambda}(\epsilon) = P_0$, for $0\le \epsilon \le \gamma$.
\item $P_B(\epsilon)\, P_C(0) = P_C(0)$, for $B \subset C$ and $\epsilon \ge 0$.
\item $P_B(\epsilon) = P_B(0)$ if $H_B$ has spectral gap greater than $\epsilon$.
\end{itemize}
\end{cor}
\begin{proof}
The first claim follows directly from the definition of $P_B(\epsilon)$, since $1-P_B(\epsilon)$ includes only states with energy at least $\epsilon$ in $H_B$. The second claim follows from the fact that $1-P_B(\epsilon)$ is orthogonal to $P_B(\epsilon')$ for $\epsilon'\le \epsilon$. Moreover, $P_\Lambda(\epsilon) = P_\Lambda(\epsilon') = P_0$, for all $\epsilon' \in [0,\epsilon]$, since $H_0$ has spectral gap $\gamma \ge \epsilon$. The third claim follows from the second claim and $P_B(0) P_C(0) = P_C(0)$, which is true for frustration-free Hamiltonians since $H_C P_C(0) = 0$ implies $H_B P_C(0) = 0$. In particular, we get: $P_B(\epsilon) P_C(0) = (P_B(\epsilon) P_B(0))\, P_C(0) = P_C(0)$. The fourth claim follows from the definition of $P_B(\epsilon)$ and the spectral decomposition of $H_B$.
\end{proof}

We are now ready to introduce the main assumption needed for the stability of frustration-free systems. It is a condition on the degree of topological order for local groundstates, so we shall call it {\it Local Topological Quantum Order [Local-TQO]}. An equivalent way of viewing Local-TQO is given in Corollary \ref{cor:local_same}, where it is shown that the condition corresponds to {\it local-indistinguishability} of the groundstates of local Hamiltonians. In other words, the Local-TQO condition formalizes the notion that different groundstates arise from the application of different boundary conditions on a common low-energy subspace.


\subsection{Topological Order Assumption.}
\begin{defn}{[}Local-TQO{]}\label{def:local_tqo}
Let $A = b_u(r)$ with $r\leq L^{*} < L$ and $O_{A}$ be any bounded operator supported on $A$. Set $A(\ell) = b_u(r+\ell)$ and define 
\[c_{\ell}(O_{A}) := \Tr(P_{A(\ell)} O_A)/\Tr P_{A(\ell)}, \quad c(O_{A}) := \Tr(P_0 O_A)/\Tr P_0.
\]
Then, for each fixed $1\le \ell \le L-r$, we have: 
\[
\left\Vert P_{A(\ell)} O_{A} P_{A(\ell)} -c_{\ell}\left(O_{A}\right) P_{A(\ell)}\right\Vert \leq\|O_{A}\|\,\Delta_{0}(\ell),
\]
where $\Delta_{0}(\ell)$ is a decaying function of $\ell$.
\end{defn}

The cut-off parameter for topological order, $L^*$, may be thought of as the {\it code distance} of the encoding one may define on the (degenerate) groundstate space of $H_0$. Since we are interested in the stability of the spectral gap of $H_0$, we expect that $L^*$ should scale with $L$. To see this, note that if the \emph{Local-TQO} condition fails after some fixed distance, then there must be operators $O_A$, with $A$ of increasing diameter and $\|O_A\|=1$, such that $\braket{\Psi_0^1}{O_A \Psi_0^1}-\braket{\Psi_0^2}{O_A \Psi_0^2} \ge c_0$, for some constant $c_0 > 0$ and $\{\ket{\Psi_0^1},\ket{\Psi_0^2}\}$ groundstates of $H_0$. If we assume that the number of these operators $O_A$, for a given size of their support $A$, scales with $L$, then it follows that perturbing $H_0$ by $\sum_A O_A$, would create a macroscopic energy gap within the groundstate subspace, thus destroying any notion of stability and macroscopic order. This is exactly what happens in the case of the two-dimensional Ising Hamiltonian $H_0 = -\sum_{p\sim q} \sigma^z_p \sigma^z_q$, under a uniform magnetic field $J \sum_p \sigma^z_p$. In particular, the two-fold degenerate groundstate space fails to satisfy \emph{Local-TQO} at the most basic level, since any local operator $\sigma^z_p, \, p \in \Z^2$ can distinguish the groundstates (one has eigenvalue $+1$ and the other $-1$.) 

One may wonder at this point why it is necessary to demand that low-energy states are locally indistinguishable at all scales, and not just at the global scale of the system. In other words, why do we care if the groundstates of the local Hamiltonians are also topologically ordered? The reason is best exemplified with the following counterexample: Let $H_0 = -\sigma^z_0 - \sum_{p\sim q} \sigma^z_p \sigma^z_q$, where $p,q \in \Z^2$. This is a gapped Hamiltonian with a unique groundstate $\ket{0\cdots 0}$, satisfying topological order at the global level (trivially, since it is the unique groundstate). It is easy to see that a perturbation of the form $J \sum_{p\in \Lambda} \sigma^z_p$, with $J>0$, will close the gap between the states $\ket{000\cdots 0}$ and $\ket{111\cdots 1}$, even for $J \sim 1/|\Lambda|$. Since \emph{Local-TQO} fails for any region that does not include the origin, this implies that some form of \emph{Local-TQO} is necessary to protect the spectral gap from collapsing. The above examples validate our intuition that Hamiltonians are {\it unstable} because local ``order parameters'' can act as perturbations to either \emph{open the gap} between groundstates, or \emph{close the gap} between global groundstates and excited states which have low energy locally. In that sense, the above condition implies that states that minimize the energy of the Hamiltonian as more and more interactions are added, differ only near the boundary of the region we are considering.

The \emph{Local-TQO} condition defined above, combines and generalizes assumptions $TQO$-$1,2$ of the stability result for the commuting case \cite{bravyi:2010a,bravyi:2010b}. It is worth noting that the above condition is a property of the local, low-energy subspaces. In fact, it is possible in many cases (e.g. Hamiltonians with a unique Matrix Product State as their groundstate, or injective PEPS~\cite{P-GVCW:2008}, \cite{schuch:2010}, \cite{schuch:2010a}) to modify a frustration-free Hamiltonian $H_0$ to a different Hamiltonian $H_1$, satisfying {\it Local-TQO} without changing the global ground state subspace, or closing the gap. In such cases, $H_0$ and $H_1$ are both gapped, frustration-free and share the same ground state space, so $H_s = (1-s)\, H_0 + s H_1$ is also gapped and frustration-free with the same ground state space. In other words, if $H_1$ satisfies {\it Local-TQO} and $H_0$ does not, we can apply the stability result to $H_1$ to infer properties about the groundstate of $H_0$. For a concrete example, one can look at two different {\it toric code} Hamiltonians, as discussed in \cite[Sec. 2.D]{bravyi:2010a}.

\begin{cor}\label{cor:local_tqo} Let $O_{A}$ be any bounded operator supported on the ball $A=b_u(r)$, with $r \le L^*$ and $u\in\Lambda$.
Then, \emph{Local-TQO} implies the following bounds for $1 \le \ell \le L-r$:
\begin{align}
\left|\left\Vert O_{A}P_{A(\ell)}\right\Vert -\left\Vert O_{A}P_0\right\Vert \right|&\leq 2\|O_{A}\|\,\sqrt{\Delta_{0}(\ell)}.\\
\|P_{A(\ell)}-P_{A(\ell)}P_A\| &\le \sqrt{3}\sqrt{\Delta_{0}(\ell)}.
\end{align}
 \end{cor}

Abusing notation, the projections $P_A$ and $P_{A(\ell)}$, above, correspond to $P_A(\epsilon)$ and $P_{A(\ell)}(\epsilon)$, respectively. All that is needed for the bounds to hold is that the projections satisfy the \emph{Local-TQO} condition. Of course, for $P_A = P_A(0)$ and $P_{A(\ell)}= P_{A(\ell)}(0)$, we have from Corollary~\ref{cor:local_energy} that $\|P_{A(\ell)}-P_{A(\ell)}P_A\| = 0$. The proof is given in Appendix~\ref{append:local_tqo}.

The next Corollary makes precise the notion of {\it local indistinguishability} between ground states and provides an equivalent condition to \emph{Local-TQO}.

\begin{cor}\label{cor:local_same}
Let $A=b_u(r)$, with $r \le L^*$ and $u\in\Lambda$. For any two groundstates $\ket{\psi_1(A(\ell))}, \ket{\psi_2(A(\ell))}$ of $H_{A(\ell)}$; that is, $P_{A(\ell)} \ket{\psi_i(A(\ell))} = \ket{\psi_i(A(\ell))}, \, i =1,2$, define the reduced density matrices $\rho_i(A) = \Tr_{A(\ell)\setminus A} \pure{\psi_i(A(\ell))}$. Then, \emph{Local-TQO} is equivalent to:
\begin{equation}\label{local_tqo:2}
\|\rho_1(A) - \rho_2(A)\|_1 \le 2\Delta_{0}(\ell).
\end{equation}
\end{cor}
The proof of the corollary is given in Appendix~\ref{append:local_same}.

Moreover, {\it Local-TQO} implies an area law \cite{hastings:2007} for the entanglement entropy of the ground states of $H_0$.
\begin{cor}\label{cor:area_law}
Let $A=b_u(r)$, with $r \le L^*$ and $u\in\Lambda$. Any groundstate $\ket{\Psi_0}$ of $H_0$ satisfying \emph{Local-TQO}, also satisfies an area law for the entanglement entropy of $\rho_A := \Tr_{A^c} \pure{\Psi_0}$:
\begin{equation}\label{eq:area_law}
S(\rho_A) \le (c_d \ln D) \, (1+r)^{d-1} \cdot \ell_0,
\end{equation}
where $c_d$ is a constant depending only on the dimension $d$ of the lattice, $D$ is the maximum dimension of the on-site Hilbert spaces (e.g. $D=2$ for spin-$1/2$ particles) and $\ell_0 = \min \{\ell: \Delta_0(\ell) \le \ell/(1+r) \}$.
\end{cor}

\begin{proof}
We begin by writing the Schmidt decomposition of $\ket{\Psi_0}$ over the partition of the lattice $A(\ell):\Lambda\setminus A(\ell)$, where $\ell$ will be chosen later:
\be\label{eq:schmidt_dec}
\ket{\Psi_0} = \sum_{k \ge 0} \sqrt{\sigma_k} \ket{\Psi_k(A(\ell))} \otimes \ket{\Psi_k(\Lambda\setminus A(\ell))}, \quad \braket{\Psi_k(\Lambda\setminus A(\ell))}{\Psi_l(\Lambda\setminus A(\ell))} = \delta_{k,l}, \quad \sum_{k\ge 0} \sigma_k =1, \quad \sigma_k > 0.
\ee
Note that for frustration-free Hamiltonians, the states $\{\ket{\Psi_k(A(\ell))}\}_{k\ge 0}$ in the Schmidt decomposition of $\ket{\Psi_0}$ are groundstates of the Hamiltonian $H_{A(\ell)}$. To see this, recall that $H_0 = \sum_u Q_u, \, Q_u \ge 0$ and $H_0 \ket{\Psi_0} = 0$, which implies that $\bra{\Psi_0} H_{A(\ell)} \ket{\Psi_0} = 0 = \sum_{k\ge 0} \sigma_k \bra{\Psi_k(A(\ell))} H_{A(\ell)} \ket{\Psi_k(A(\ell))}$, hence $H_{A(\ell)} \ket{\Psi_k(A(\ell))} = 0, \, k\ge 0$. Setting $\rho_k(A) = \Tr_{A(\ell)\setminus A} \pure{\Psi_k(A(\ell))}$, we get from (\ref{eq:schmidt_dec}):
\be\label{eq:rho_A}
\rho_A := \Tr_{\Lambda\setminus A} \pure{\Psi_0} = \sum_{k\ge0} \sigma_k \, \rho_k(A) = \rho_0(A) + \sum_{k\ge0} \sigma_k \, (\rho_k(A)-\rho_0(A)).
\ee
Corollary~\ref{cor:local_same} and (\ref{eq:rho_A}) give: $\|\rho_A - \rho_0(A)\|_1 \le 2 \Delta_0(\ell)$. Combined with the bound $S(\rho_o(A)) \le (|A(\ell)|-|A|) \, \ln D \le c'_d\, (1+r)^{d-1} \ell \, \ln D$, for $\ell \le r$ and $c'_d$ a constant depending only on $d$, we have:
\be
S(\rho_A) \le S(\rho_0(A)) + |S(\rho_A) - S(\rho_0(A))| \le c'_d\, \ln D\, (1+r)^{d-1} \ell + c''_d \, \ln D \,\Delta_0(\ell) \, (1+r)^d + H(\Delta_0(\ell)),
\ee
where we used the Fannes-Audenaert inequality \cite{fannes,audenaert} for the entropy: $$|S(\rho_A) - S(\rho_0(A))| \le (\|\rho_A-\rho_0(A)\|_1/2) \, \ln D^{|A|} + H(\|\rho_A-\rho_0(A)\|_1/2),$$ with $H(p) = -p \ln p - (1-p) \ln (1-p), \, p \in [0,1]$. Setting $\ell_0 = \min \{\ell: \Delta_0(\ell) \le \ell/(1+r) \}$ and $c_d = \max\{c'_d,c''_d\} +1$, we get the desired bound, noting that $H(\Delta_0(\ell_0)) \le \ln 2$.
\end{proof}
We note that for \textit{stabilizer} Hamiltonians, like the toric code, the \emph{Local-TQO} decay $\Delta_0(\ell)$ vanishes after a constant length $\ell_0$; in other words, $\Delta_0(\ell) = 0, \, \ell \ge \ell_0$. This implies that for such Hamiltonians, the groundstates satisfy an area law exactly. On the other hand, if the decay $\Delta_0(\ell)$ is exponential in $\ell$, then there is a logarithmic divergence ($\ell_0 \sim \ln \,(1+r)$) in the area law, which is natural for gapless systems (the above proof made no use of a spectral gap.)

\subsection{Local Gap Assumption}
The second assumption imposes a condition on the spectral gaps of the local Hamiltonians $H_{b_u(r)}$.
\begin{defn}{[}Local-Gap{]}\label{def:local_gap} 
We say that $H_0$ is locally gapped with gap $\gamma(r) > 0$, iff for each $u \in \Lambda$ and $r \ge 0$, $P_{b_u(r)}(\gamma(r)) = P_{b_u(r)}$. When $\gamma(r)$ decays at most polynomially in $r$, we say that $H_0$ satisfies the \emph{Local-Gap} condition.
\end{defn}
Below, we present an example of a gapped, translationally-invariant, frustration-free Hamiltonian with \emph{Local-TQO}, that is unstable under vanishingly small perturbations:
\begin{example}
Let $H_{2N} = \sum_{k=1}^{2N} H_{k,k+1}$, where $H_{k,k+1} = 1/(3N) \pure{0}_k\otimes\pure{1}_{k+1} + \pure{1}_k\otimes\pure{1}_{k+1}+ \pure{0}_k\otimes\pure{0}_{k+1},$ for $k$ even and $H_{k,k+1} = 1/(3N) \pure{1}_k\otimes\pure{0}_{k+1} + \pure{1}_k\otimes\pure{1}_{k+1}+ \pure{0}_k\otimes\pure{0}_{k+1},$ for $k$ odd. There is a unique groundstate $\ket{0101\ldots 01},$ with energy $0$ and spectral gap $2/3$ to the first excited state $\ket{1010\ldots 10}$. Note that since all eigenstates of $H_{2N}$ can be written as product states over $\ket{0},\ket{1}$, any state with consecutive $0$'s, or $1$'s has energy at least $1$. It should be obvious that this Hamiltonian satisfies frustration-freeness and \emph{Local-TQO} (trivially, since the groundstate is unique and is a product state), but is unstable under the perturbation $V = -\sum_{k=1}^{2N} V_k$, where $V_k = 1/(3N) \pure{0}_k\otimes\pure{1}_{k+1}$, for $k$ even and $V_k = 1/(3N) \pure{1}_k\otimes\pure{0}_{k+1}$, for $k$ odd. Interestingly, the Hamiltonian $H_{2N}+V$ is frustration-free, has degenerate ground state space and satisfies the \emph{Local-Gap} condition, but not the \emph{Local-TQO} assumption, since there exist local operators, like $\pure{0}\otimes\pure{1}$, that distinguish the two ground states $\{\ket{010101\ldots 01}, \ket{101010\ldots 10}\}$. In other words, the two assumptions for stability seem to be independent and dual to each other. Moreover, this example suggests that, in some further refined form, both conditions are necessary for stability.
\end{example}
Such violations to the \emph{Local-Gap} condition are not obvious for Hamiltonians that are sums of non-vanishing terms. In fact, to the authors' knowledge, it remains an open problem whether it is possible to construct examples of gapped, frustration-free Hamiltonians that are sums of local projections and do not satisfy the {\it Local-Gap} condition. Moreover, we expect that the \emph{Local-Gap} condition can be merged with a stronger version of the \emph{Local-TQO} condition, where the projections $P_{A(\ell)}$ are substituted with $P_{A(\ell)}(\epsilon)$, for some $\epsilon < \gamma$. Furthermore, we expect that using the powerful \textit{detectability lemma} \cite{aharonov:2010}, one may be able to remove the \emph{Local-Gap} assumption, when the Hamiltonian interactions can be rescaled to projections in a system-size independent way.

\section{Transforming the perturbed Hamiltonian}\label{sec:perturbation}
The main goal of this section is to transform a $(J,f_0)$-perturbation $V$, into a $(J,w)$-perturbation $W$ with similar decay, such that: $W_u(r) P_{b_u(r)} = 0$ and for some unitary $U$ we have, $U^{\dagger} (H_0 + V) U = H_0 + W + \Delta + E \one,$ where $\|\Delta\| \rightarrow 0$ as $L \rightarrow \infty$ and $E = \Tr(U P_0 U^{\dagger} (H_0 + V))/\Tr P_0$. Recall that $P_{b_u(r)}$ is the ground state projector of the local, frustration-free Hamiltonian $H_{b_u(r)}$.

\subsection{Making the global ground states commute quasi-locally.}
In this section, we re-derive a version of Lemma $7$ in ~\cite{bravyi:2010b}, which will be sufficient for our purposes. The content of this Lemma can be summarized as the decomposition of any gapped Hamiltonian into a sum of quasi-local interactions that commute with the groundstate subspace. Note that any local, frustration-free Hamiltonian satisfies this condition (with strict locality), but once we add a perturbation, the new terms may no longer commute with the groundstate subspace.
\begin{lem}\label{lemma:global}
Let $H_s$ be the family of gapped Hamiltonians given in Definition~\ref{defn:path}. Then, there exists a linear operator $\F$, with the following properties:
\begin{enumerate}[(i)]
\item $\F(H_s) = H_s$ and, hence, $\Fz(H_0) = H_0$, 
\item $[\F(O_u(r)), P_0(s)]=0$, for any operator $O_u(r)$ supported on $b_u(r)$,  $u \in \Lambda$,
\item $\|\F(O_u(r))\| \le \|O_u(r)\|$,
\item $\F(O_u(r)) = \sum_{r' \ge 0} \F(r'; O_u(r))$, with $\F(r';O_u(r))$ supported on $b_u(r+r')$ and $\|\F(r'; O_u(r))\| \le \|O_u(r)\| \, f_0^{(1)}(r')$, for a rapidly decaying function $f_0^{(1)}$.
\item $\F(Q_u) - \Fz(Q_u) = \sum_{r \ge 2} X^{(1)}_u(r)$, with $X^{(1)}_u(r)$ supported on $b_u(r)$ and $\|X^{(1)}_u(r)\| \le sJ\, \|Q_u\| \,f^{(1)}_1(r)$, for a rapidly decaying function $f_1^{(1)}$.
\end{enumerate}
\end{lem}

\begin{proof}
Define the operator $\F$ as follows: 
\be\label{F:operator}
\F(O_u(r)) = \int_{-\infty}^{\infty} dt\, w_{\gamma'}(t) \, \tau_t^{H_s}(O_u(r)),
\ee
where $\tau_t^{B}(A) = e^{itB} A e^{-itB}$ and $w_{\gamma'}(t)$ is the filter function studied in~\cite{BMNS:2011,hastings:2010}, with the following properties:
\begin{enumerate}
\item Compact Fourier Transform \cite{I1934}: $\hat{w}_{\gamma'}(\omega) = 0,$ for $|\omega| \ge \gamma'$ and $\hat{w}_{\gamma'}(0) = 1$.
\item Rapid decay: $0 \le w_{\gamma'}(t) \le 2e^2\gamma' (\gamma' |t|) \exp\left(-\frac{2}{7} \frac{\gamma' |t|}{\ln^2 (\gamma' |t|)}\right),$ for $\gamma' |t| \ge e^{1/\sqrt{2}}$ and $0 \le w_{\gamma'}(t) \le \gamma'/\pi$, for all $t \in \R$.
\end{enumerate}
Note that, we immediately get:
\be
\F(H_s) = \int_{-\infty}^{\infty} dt\, w_{\gamma'}(t) \, e^{itH_s} (H_s) e^{-itH_s} = \left(\int_{-\infty}^{\infty} dt\, w_{\gamma'}(t)\right) H_s = H_s,
\ee
since $\int_{-\infty}^{\infty} dt\, w_{\gamma'}(t) = \hat{w}_{\gamma'}(0) = 1.$
Moreover, the new interactions commute with $P_0(s)$, since for any ground state $\ket{\psi_0(s)}$ of $H_s$ and any eigenstate $\ket{\psi_n(s)}$ orthogonal to $P_0(s)$ with eigenvalue $E_n(s)$, we have:
\be
\bra{\psi_n(s)} \F(O_u(r)) \ket{\psi_0(s)} = \hat{w}_{\gamma'}(E_{n}(s)-E_0(s)) \bra{\psi_n(s)} O_u(r) \ket{\psi_0(s)} = 0,
\ee
where $E_n(s)-E_0(s) \ge \gamma'$, by assumption. This implies that: $$(\one-P_0(s)) \F(O_u(r)) P_0(s) = 0 = P_0(s) \F(O_u(r)) (1-P_0(s)),$$ hence $\F(O_u(r)) P_0(s) = P_0(s) \F(O_u(r)) P_0(s) = P_0(s) \F(O_u(r)),$ as desired. 
We also have by the triangle inequality:
\be
\|\F(O_u(r))\| \le \int_{-\infty}^{\infty} dt\, w_{\gamma'}(t) \, \|e^{itH_s} (O_u(r)) e^{-itH_s}\| = \|O_u(r)\|.
\ee
Finally, the quasi-locality of $\F(O_u(r))$ follows from writing $\F(O_u(r)) = \sum_{r' \ge 0} \F(r';O_u(r)),$ where for $r' \ge 1$:
\be
\F(r';O_u(r))= \int_{-\infty}^{\infty} dt\, w_{\gamma'}(t) \, \left(\tau_t^{H^u_s(r+r')}\big(O_u(r)\big)-\tau_t^{H^u_s(r+r'-1)}\big(O_u(r)\big)\right),
\ee
and
\be
\F(0;O_u(r))= \int_{-\infty}^{\infty} dt\, w_{\gamma'}(t) \, \tau_t^{H^u_s(r)}\big(O_u(r)\big),
\ee
with the Hamiltonians in the exponent defined as follows, for $0 \le q \le L$: 
\be
H^u_s(q) = \sum_{v: b_v(1) \subset b_u(q)} Q_v + \sum_{(k,v): b_v(k) \subset b_u(q)} s\cdot V_v(k).
\ee
Note that $H^u_s(L) = H_s$ and that $\F(r';O_u(r))$ has support on $b_u(r+r')$ and norm decaying rapidly in $r'$. The decay follows directly from the assumption on the almost-exponential decay of $w_{\gamma'}(t)$ for $|t|$ large enough and, rapidly decaying Lieb-Robinson bounds in $r'$, on the norm of the commutator:
$$\big[H^u_s(r+r')-H^u_s(r+r'-1), \tau_t^{H^u_s(r+r')}\big(O_u(r)\big)\big],$$ for small enough $|t|$. 
In particular, note that for $$\delta_{s,t}(r';O_u(r)) =  \Big\|\tau_t^{H^u_s(r+r')}\big(O_u(r)\big)-\tau_t^{H^u_s(r+r'-1)}\big(O_u(r)\big)\Big\|,$$ we have:
\begin{align}
&\delta_{s,t}(r';O_u(r)) =  \left\|\tau_{-t}^{H^u_s(r+r'-1)}\big(\tau_t^{H^u_s(r+r')}\big(O_u(r)\big)\big)-O_u(r)\right\|\\
&\le \int_0^{|t|} \big\|\big[H^u_s(r+r')-H^u_s(r+r'-1), \tau_{t'}^{H^u_s(r+r')}\big(O_u(r)\big)\big]\big\| \, dt'.
\end{align}
Defining $I_q(u) = \{Z\subset b_u(q): Z \cap (b_{u}(q)\setminus b_{u}(q-1)) \neq \emptyset\}$, we continue the above estimate as follows:
\begin{align}
\delta_{s,t}(r';O_u(r)) &\le \sum_{b_v(1) \in I_{r+r'}(u)} \int_0^{|t|} \big\|\big[Q_v, \tau_{t'}^{H^u_s(r+r')}\big(O_u(r)\big)\big]\big\| \, dt'\\
&+ \sum_{b_v(k) \in I_{r+r'}(u)} \int_0^{|t|} s\, \big\|\big[\,V_v(k), \tau_{t'}^{H^u_s(r+r')}\big(O_u(r)\big)\big]\big\| \, dt'.
\end{align}
Note that the norm of the commutator $\big\|\big[Q_v, \tau_{t'}^{H^u_s(r+r')}\big(O_u(r)\big)\big]\big\|$ can be bounded directly using Lieb-Robinson bounds for the evolution $\tau_{t'}^{H^u_s(r+r')}$, with rapid decay in the distance between the supports of $b_v(1) \in I_{r+r'}(u)$ and $b_u(r)$, which is at least $r'-4$. Similarly, $\big\|\big[\,V_v(k), \tau_{t'}^{H^u_s(r+r')}\big(O_u(r)\big)\big]\big\|$ will decay rapidly in the distance between $b_v(k) \in I_{r+r'}(u)$ and $b_u(r)$, which is at least $r'-2k$. Of course, for $k \ge r'/2$ we may use the simple bound $\big\|\big[\,V_v(k), \tau_{t'}^{H^u_s(r+r')}\big(O_u(r)\big)\big]\big\| \le 2\|V_v(k)\|\|O_u(r)\| \le 2\,f_0(k)\, \|O_u(r)\|$.
On the other hand, we also have the trivial bound for all $t \in \R$:
\be
\delta_{s,t}(r';O_u(r)) \le \big\|\tau_t^{H^u_s(r+r')}\big(O_u(r)\big)\big\| + \big\|\tau_t^{H^u_s(r+r'-1)}\big(O_u(r)\big)\big\|= 2 \|O_u(r)\|.
\ee
Combining the above bounds with the usual technique of bounding the integral: $$\int_{-\infty}^{\infty} f(t)\cdot g(t) \,dt \le \|g\| \int_{-T}^{T} |f(t)|\, dt + \|f\| \int_{|t| \ge T} |g(t)| \,dt,$$ where $\|f\|$ is the supremum of $f$ over the appropriate domain, gives us the fastest decay for the optimal choice of the cut-off time $T$. In our case, $f(t) = \delta_{s,t}(r';O_u(r))$ and $g(t) = w_{\gamma'}(t)$.

Now, for the final statement of the Lemma, we note that Duhamel's differentiation formula implies:
\begin{align}
&\F(Q_u)-\Fz(Q_u) = \int_{-\infty}^{\infty} dt\, w_{\gamma'}(t) \, \int_0^s \partial_{s'} \tau_t^{H_{s'}}(Q_u) \,ds' = \\
&\int_{-\infty}^{\infty} dt\, w_{\gamma'}(t) \, \int_0^s i \Big[\int_0^t \tau_{t'}^{H_{s'}}(V) \, dt' ,\tau_t^{H_{s'}}(Q_u)\Big] \,ds' =\\
&\int_0^s \,ds' \int_{-\infty}^{\infty} dt\, w_{\gamma'}(t) \, \tau_t^{H_{s'}}\Big(\int_0^t i\big[\tau_{-t'}^{H_{s'}}(V),Q_u\big] \, dt' \Big),
\end{align}
where we performed a change of variables and used the linearity of the commutator to get the last line.
Note that in the last line above, we have an evolution $\tau_t^{H_{s'}}$ applied to an integral involving the commutator:
$$\big[\tau_{-t'}^{H_{s'}}(V), Q_u\big] = \sum_{(v,k)} \big[\tau_{-t'}^{H_{s'}}(V_v(k)), Q_u\big].$$
It should be clear at this point that by applying the ``telescoping sum'' procedure, described earlier in this proof, to the evolution $\tau_{-t'}^{H_{s'}}$, each commutator in the above sum is itself a sum of terms supported on larger and larger sets, with decreasing norm. To see this, note that we may bound the norm of each such term trivially:
\be
\big\|\big[\tau_{-t'}^{H^v_{s'}(k+k')}(V_v(k))-\tau_{-t'}^{H^v_{s'}(k+k'-1)}(V_v(k)), Q_u\big]\big\| \le 2\, \delta_{s',-t'}(k';V_v(k)) \, \|Q_u\|,
\ee
using the appropriate bounds on $\delta_{s',-t'}(k';V_v(k))$, as described above.
In all cases, we note that the commutator vanishes for $k+k' \le d(v,u)-3$, where $d(v,u)$ is the distance between the sites $u,v \in \Lambda$, so we only need to consider $k' \ge \max\{0,d(u,v)-k-2\}$. Finally, after localizing the terms inside the evolution $\tau_t^{H_{s'}}$, we perform another ``telescoping sum'' decomposition on this final evolution, arriving at the final decomposition of strength $J$ (because of the norm of the terms $V_v(k)$) and decay given by a function $f^{(1)}_1$, whose decay depends implicitly on the decay of $w_{\gamma'}$ and Lieb-Robinson bounds corresponding to the unitary evolutions under Hamiltonians with interactions coming from $Q_u, \, u \in \Lambda$ and $V_v(k), \, v \in \Lambda, \, k \ge 0$. The linear dependence of the bound on $s$ follows from the triangle inequality: $$\Big\|\int_0^s ds' \,A(s')\Big\| \le s \cdot \sup_{s'\in[0,s]} \|A(s')\|.$$
\end{proof}

\subsection{Absorbing the perturbation.}
In this subsection, we focus on unitarily transforming the Hamiltonian $H_s = H_0+sV$, satisfying the conditions of Definition~\ref{defn:path}, into a Hamiltonian $H'_s= H_0+V'_s + E_s \one$, with the new perturbation $V'_s$ satisfying certain stronger conditions with respect to the local ground states of $H_0$, as discussed at the beginning of this section. Since the whole process is a unitary transformation combined with a global energy shift, the spectral gap of $H_s$ is identical to that of $H_0+V'_s$.
The transformation is done in four steps:
\begin{enumerate}
\item We transform the terms $\F^u := \F(Q_u)+s\,\F(V_u)$ of $H_s = \sum_u \F^u$, satisfying $[\F^u, P_0(s)] = 0$, into terms $X_u = U^\dagger(s) \F^u U(s) -  \Fz(Q_u),$ satisfying: $$[X_u, P_0] = U^\dagger(s) [\F^u, P_0(s)] U(s) - [\Fz(Q_u),P_0]= 0,$$ where the unitary $U(s)$ satisfies $U(s) P_0 U^\dagger(s) = P_0(s)$. At this point, the Hamiltonian $H_s$ has been unitarily transformed into $H'_s := U^\dagger(s) H_s U(s) = H_0 + \sum_u X_u$.
\item We decompose $X_u = X^{(1)}_u + X^{(2)}_u + X^{(3)}_u$, where:
\begin{align*} 
X^{(1)}_u = \F(Q_u) - \Fz(Q_u),\quad
X^{(2)}_u =U^\dagger(s) \F(Q_u) U(s) -  \F(Q_u), \quad
X^{(3)}_u = s\, U^\dagger(s) \F(V_u) U(s),
\end{align*}
and show that each term is a strength $J$ perturbation with decay given by functions $f^{(1)}_1$, $f^{(2)}_1$ and $f^{(3)}_1$, respectively, using Lemma~\ref{lemma:global} and properties of the unitary $U(s)$~\cite{BMNS:2011,hastings:2010}. 
\item Using $[X_u,P_0]=0$, we rewrite $X_u = (1-P_0) \X_u (1-P_0) + \Delta_u + c(X_u)\one$,
where $\X_u= X_u - c(X_u) \one$, $\Delta_u = P_0 \X_u P_0$ and $c(X_u)= \Tr (P_0 X_u)/\Tr P_0$. We show $\|\Delta_u\|$ decays rapidly in $L^*$, using \emph{Local-TQO} and the decomposition of $X_u$ from the previous step.
\item For a $(J,f_1)$ perturbation $X_u$ satisfying $[X_u,P_0]=0$, we prove that $W_u := (1-P_0) \X_u (1-P_0)$ is a $(J,w)$ perturbation, satisfying $W_u(r) P_{b_u(r)}=0,$ with the decay $w(r)$ closely related to $f_1(r)$.
\end{enumerate} 
Note that throughout the transformation process, the perturbations $X_u, W_u$ and $\Delta_u$ depend implicitly on the parameter $s$.
Nevertheless, the locality and bounds on the strength of the perturbations are independent of $s$ (some bounds have a linear dependence on $s$, but since $s\in[0,1]$, this makes the bounds only stronger.)

We state the end result of the above process as a Proposition:
\begin{prop}\label{prop:local_decomposition}
Let $H_s$ be the family of Hamiltonians considered in Definition~\ref{defn:path}. Then, there exist unitaries $U(s)$, such that:
$U^\dagger(s) H_s U(s) = H_0 + W + \Delta + E_0(s) \one,$
with $W$ a $(J,w)$-perturbation satisfying, $W_u(r) P_{b_u(r)}=0$, $\|\Delta\|$ rapidly decaying in $L$ and $E_0(s)=\Tr(H_s P_0(s))/\Tr P_0(s)$.
\end{prop}
\begin{proof}
Using Lemma~\ref{lemma:global}, we see that the first step implies $H'_s := U^\dagger(s) H_s U(s) = U^\dagger(s) \F(H_s) U(s) = \Fz(H_0) + \sum_u X_u= H_0 + \sum_u X_u$, where $[X_u,P_0]=0$. The second and third step imply: $$H'_s = H_0 + \sum_u (1-P_0) \X_u (1-P_0) + \Delta + c(X) \one,$$ for $\Delta = \sum_u \Delta_u$ a vanishingly small perturbation and $X= \sum_u X_u$. Noting that $X = U^\dagger(s) H_s U(s) - H_0$ and $P_0(s) = U(s) P_0 U^\dagger(s)$, the cyclicity of the trace implies:
$$c(X) = \Tr(H_s U(s) P_0 U^\dagger(s))/\Tr P_0 -\Tr(H_0 P_0)/\Tr P_0= \Tr \big(H_s P_0(s)\big)/\Tr P_0(s),$$ so 
$H'_s = H_0 + (1-P_0) \X (1-P_0) + \Delta + E_0(s) \one.$ 
Finally, setting $W = \sum_u W_u$ and using the second and fourth steps, we get the desired result.
\end{proof}
We focus, now, on proving the last three steps above, since the first step follows immediately from the properties of $\F$ and the quasi-adiabatic evolution $U(s)$ (see, e.g. \cite[Prop. 2.4]{BMNS:2011}.) For the second step above, we make extensive use of results obtained in~\cite{BMNS:2011,hastings:2010}. For an overview of the ideas behind Hastings' quasi-adiabatic evolution $U(s)$ and its many applications to problems in condensed matter physics and quantum information theory, we refer the reader to~\cite{hastings:2010}.
Below, we reference results from~\cite{BMNS:2011}, from which one may extract explicit bounds on the decay of $X^{(2)}_u$ and $X^{(3)}_u$. The properties of $X^{(1)}_u$ follow directly from Lemma~\ref{lemma:global}.
\begin{lem}\label{lemma:local}
Let $H_s$ be the family of gapped Hamiltonians described in Definition~\ref{defn:path} and define $\G(O_u(r)) := \alpha_s^{\Lambda} (O_u(r))$, where $\alpha_s^{\Lambda}$ is defined in~\cite{BMNS:2011} as the spectral flow corresponding to $H_s$ and $O_u(r)$ is an operator supported on $b_u(r)$. Then, 
\be\label{flow_decomposition}
\G(O_u(r)) = \sum_{r' = r_0}^{L-r} \G(r';O_u(r)), \quad r_0 \ge 0,
\ee
with $\G(r';O_u(r))$ supported on $b_u(r+r')$ and 
\begin{align}\label{flow_bound}
\|\G(r';O_u(r))\| &\le J \, \|O_u(r)\| \, g(r'), \quad r' > r_0, \\
\|\G(r_0;O_u(r))-O_u(r)\| &\le sJ \, \|O_u(r)\| \, g(r_0), \quad \|\G(r_0;O_u(r))\|= \|O_u(r)\|, \nonumber
\end{align}
for a rapidly decaying function $g$.
\end{lem}

\begin{proof}
We introduce the decomposition $\G(O_u(r)) = \sum_{r' = r_0}^{L-r} \G(r';O_u(r))$:
\begin{align}
&\G(r';O_u(r)) =\alpha_s^{\Lambda_{r'}}(O_u(r))-\alpha_s^{\Lambda_{r'-1}}(O_u(r)),\quad r'\ge r_0+1,\\
&\G(r_0;O_u(r)) = \alpha_s^{\Lambda_{r_0}}(O_u(r))\quad r_0 \ge 0,
\end{align}
where we have defined $\Lambda_{n} = b_u(r+n), \, n \ge 0$ and used Eq. (4.42-4.46) of ~\cite{BMNS:2011} to define $\alpha_s^Y(A) = U^\dagger_{Y}(s) \, A\, U_{Y}(s)$:
\be
\partial_s U_{Y}(s) = i \,D_Y(s)\, U_{Y}(s), \quad U_{Y}(0) =\one, \quad D_Y(s) = \sum_{Z \subset Y} \Psi_{\Lambda}(Z,s).
\ee
First, notice that since each term $\Psi_{\Lambda}(Z,s)$ is supported on $Z$ by definition (see \cite{BMNS:2011}), the terms $\G(r';O_u(r))$ have support on $b_u(r+r')$, as desired. Moreover, since $\G(r_0;O_u(r)) = \alpha_s^{\Lambda_{r_0}}(O_u(r)) = U^\dagger_{\Lambda_{r_0}}(s) \, O_u(r)\, U_{\Lambda_{r_0}}(s)$, we immediately get $\|\G(r_0;O_u(r))\|= \|O_u(r)\|$.
To bound the norm of $\G(r';O_u(r))$ for $r' > r_0$, we need to bound $\|\alpha_s^{\Lambda_{r'}}(O_u(r)) - \alpha_s^{\Lambda_{r'-1}}(O_u(r))\|$. As we show in Appendix~\ref{lr-estimates2}, the norm $\|\alpha_s^{\Lambda_{r'}}(O_u(r)) - \alpha_s^{\Lambda_{r'-1}}(O_u(r))\|$ is bounded above by:
\be\label{bound:Xu3}
J \, \|O_u(r)\| \, g(r') ,
\ee
for a rapidly decaying function $g$. A similar estimate shows that:
\be
\|\G(r_0;O_u(r))-O_u(r)\| \le \int_0^s ds'\, \|[D_{\Lambda_{r_0}}(s'), \alpha_s^{\Lambda_{r_0}}(O_u(r))]\| \le sJ \, \|O_u(r)\| \, g(r_0).
\ee
\end{proof}

The third step makes use of the \emph{Local-TQO} condition, allowing us to extract the groundstate energy of $H_s$ as a single quantity $E_0(s)$, up to a small correction. In other words, the following Lemma asserts the vanishingly small splitting of the groundstate subspace for the gapped Hamiltonians $H_s$.
\begin{lem}\label{lemma:D_u}
Let $X_u$ be a $(J,f_1)$-perturbation centered on $u\in\Lambda$. Then, Assumption~\ref{def:local_tqo} implies for $\Delta_u=P_0 \X_u P_0$, $\X_u = X_u-c(X_u) \one$, $c(X_u) = \Tr (X_uP_0) / \Tr P_0$:
\be
\|\Delta_u\| \le J \Big(\sum_{r=0}^{L^{*}}f_1(r) \Delta_{0}(L-r) +2 \sum_{r>L^{*}}^L f_1(r)\Big).\label{decay2}
\ee
\end{lem}
\begin{proof}
By assumption, $X_u = \sum_{r\ge 0} X_u(r)$, with $X_u(r)$ supported on $b_u(r)$ and $\|X_u(r)\| \le J\, f_1(r)$.
We introduce an energy shift by redefining the local terms:
\be
\X_{u}(r) := X_{u}(r)-c(X_{u}(r))\one, \quad c(X_u(r)) = \Tr (P_0 X_u(r))/\Tr P_0.
\ee
From \emph{Local-TQO} we know that for $r\le L^{*}$, 
\be\label{bnd:1}
\|P_0\X_{u}(r)P_0\|\le \|X_{u}(r)\| \, \Delta_{0}(L-r).
\ee
On the other hand, for $r > L^*$, we have:
\be\label{bnd:2}
\|P_0\X_{u}(r)P_0\|\le \|X_u(r)\| + |c(X_u(r))| \le 2J f_1(r).
\ee
Using (\ref{bnd:1}) and (\ref{bnd:2}), we get the desired bound: 
\begin{align}
\|P_0\X_u P_0\| \leq J \Big(\sum_{r=0}^{L^{*}}f_1(r) \Delta_{0}(L-r) +2 \sum_{r>L^{*}}^L f_1(r)\Big).
\end{align}
\end{proof}

We now move to the last step of our transformation.
\begin{lem} Consider a frustration-free Hamiltonian $H_0$ and a perturbation $X_u$ with strength $J$ and decay $f_1(r)$, satisfying $[X_u,P_0]=0$. Set $W_u = (1-P_0) \X_u (1-P_0)$, where $\X_u = X_u - c(X_u) \one$. Then, $W_u$ is a strength $J$ perturbation with decay $w(r)$ given below and $W_u(r) P_{b_u(r)}=0$.
\end{lem}
\begin{proof} 
By definition, $W_{u}P_0=0$. Moreover, using $[X_u,P_0]=0$, we have: $\X_u = W_u + \Delta_u$, where $\Delta_u = P_0 \X_u P_0$. This implies that:
\be
W_u = \sum_{r\ge 0}^{L-1} \X_u(r) + (\X_u(L)- \Delta_u).
\ee
Let us choose $\ell\leq L^{*}$. We apply Corollary $\ref{cor:local_tqo}$ to $\sum_{r=0}^{\ell}\X_{u}(r)$:
\begin{align}
&\big\|\sum_{r=0}^{\ell}\X_{u}(r)\,P_{b_u(2\ell)}\big\| \leq\big\|\sum_{r=0}^{\ell}\X_{u}(r)\,P_0\big\|+2\,\|\sum_{r=0}^{\ell}\X_{u}(r)\|\,\sqrt{\Delta_{0}(\ell)}\\
& \leq\underset{0}{\underbrace{\|W_{u}P_0\|}}+\big\|\sum_{r>\ell}^{L}\X_{u}(r)P_0\big\|+\|\Delta_u\| + 2\,\|\sum_{r=0}^{\ell}\X_{u}(r)\|\,\sqrt{\Delta_{0}(\ell)}\\
& \leq 2J\Big(\sum_{r>\ell}^{L} f_1(r) +2\sum_{r=0}^{\ell} f_1(r)\,\sqrt{\Delta_{0}(\ell)} \Big)+ \|\Delta_u\|,\label{bound1}
\end{align}
where we used $\|\X_u(r)\| \le 2\|X_u(r)\|$.
We turn our attention to the final decomposition of $W_u$ into interaction terms that are annihilated by local ground states.

Define an orthogonal unity decomposition $\one=\sum_{m=1}^{L+1}E_{m}$, as in \cite[Sec. 4]{bravyi:2010b}:
\begin{align*}
E_{m} & =P_{b_{u}(m-1)}-P_{b_{u}(m)},\quad 1\leq m\leq L+1, \quad P_{u} =\one,\quad P_{b_u(L+1)}=0.
\end{align*}
Using $\sum_{1\le p,r \le L} = \sum_{j\ge 2}^{2L} \sum_{p+r=j}$ and $\sum_{r\le p\le s}E_{p} = P_{b_u(r-1)}-P_{b_u(s)}$, we get: 
\[ W_{u}:=(1-P_0)\X_u (1-P_0) = \sum_{j\ge2}^{L-1}Y_u(j) + Y_u(L)+\sum_{2q=2}^{L} Z_u(2q-1) + Z_u(L),\]
where we defined:
\begin{align*}
&Y_u(j) :=\sum_{p+r=j} E_{p}\, \sum_{q=0}^{\ell(p,r)}\X_{u}(q)\,E_{r}, \quad 2\le j \le L-1,\\
&Y_u(L) := \sum_{p+r = L}^{2L} E_{p}\, \sum_{q=0}^{\ell(p,r)}\X_{u}(q)\,E_{r},\quad \ell(p,r) := \lfloor\max(p,r)/2\rfloor,\\
&Z_u(2q-1) := (\one-P_{b_u(2q-1)})\, \X_{u}(q)\, (\one-P_{b_u(2q-1)}), \quad 2\le 2q \le L,\\
&Z_u(L) := (1-P_0)\sum_{q > L/2} \X_u(q) (1-P_0).
\end{align*}

Note that $Y_u(j)P_{b_j(u)}=0$ and $Z_u(2q-1)P_{b_u(2q-1)}=0$. In other words, $Y_u(j)$ and $Z_u(2q-1)$ are supported on $b_j(u)$ and $b_u(2q-1)$, respectively, and are annihilated by the local ground states. The same holds for the terms $Y_u(L)$ and $Z_u(L)$.
Finally: 
\begin{align}\label{bound:decay1}
&\|Z_u(2q-1)\| \le \|\X_u(q)\| \le 2J f_1(q), \quad \|Z_u(L)\| \le 2\,J \sum_{q > L/2} f_1(q),\\
&\|Y_u(j)\| \le \sum_{p+r=j}\Big\|E_{p}\,\sum_{q=0}^{\ell(p,r)}\X_{u}(q)\,E_{r}\Big\|\\
& \le 2 \sum_{p \ge j-p} \Big\|E_{p} \sum_{q=0}^{\lfloor p/2\rfloor}  \X_{u}(q)\Big\|
    \le 2 \sum_{p \ge j-p} \Big\|P_{b_u(p-1)} \sum_{q=0}^{\lfloor p/2\rfloor}  \X_{u}(q)\Big\|\\
& \le 2J \sum_{p>\lfloor j/2\rfloor}^{j-1} \Big(\sum_{r\ge\lfloor p/2\rfloor}^{L} f_1(r) + 2\, \sum_{r=0}^{\lfloor p/2\rfloor} f_1(r) \sqrt{\Delta_{0}(p)} \Big) + j\, \|\Delta_u\|,
\end{align}
where the last line follows from (\ref{bound1}). 
To bound the norm of $Y_u(L)$, we simply sum over the bounds for the norms of $Y_u(j)$:
\be\label{bound:decay2}
\|Y_u(L)\| \le 2J \sum_{j \ge L}^{2L} \sum_{p>\lfloor j/2\rfloor}^{j-1} \Big(\sum_{r\ge\lfloor p/2\rfloor}^{L} f_1(r) + 2\,\sum_{r=0}^{\lfloor p/2\rfloor} f_1(r) \sqrt{\Delta_{0}(p)} \Big) + 2L^2\, \|\Delta_u\|.
\ee
We define the function $w(r)$ as follows:
\begin{align} \label{decay1}
&w(r) = \|Y_u(r)\|, \quad r \equiv 0 \pmod 2, \quad r < L,\\
&w(r) = \|Y_u(r)\|+\|Z_u(r)\|, \quad r \equiv 1 \pmod 2\quad r < L,\nonumber\\
&w(L) = \|Y_u(L)\|+\|Z_u(L)\|.\nonumber
\end{align}
For $\Delta_{0}(r)$ and $f_1(r)$ decaying rapidly, we see from (\ref{bound:decay1}-\ref{bound:decay2}) and Lemma~\ref{lemma:D_u} that $w(r)$ is also rapidly decaying.
\end{proof} 

\section{Proof of Relative Boundedness}\label{sec:relative_bound}
We now generalize the proof of Proposition $1$ in \cite[Sec. 5]{bravyi:2010b}. 
\begin{prop}\label{prop:relative_bound}
Let $W$ be a $(J,w)$-perturbation satisfying: $W_u(r) P_{b_u(r)}=0$. For $H_0$ a frustration-free, locally-gapped Hamiltonian on $\Lambda \subset \Z^d$, with gap decay $\gamma(r)$, there exists a constant 
\be
c := C_d \sum_{k=1}^{M} \frac{r_k^d}{\gamma(r_k)} \sum_{r > r_{k-1}}^{r_k} w(r), \quad 0=r_0<r_1 \le \ldots \le r_M = L,
\ee
where $C_d$ is the volume of the unit ball in $\Z^d$, such that for arbitrary states $\ket{\psi}$:
\be
|\bra{\psi} W \ket{\psi}| \le c\cdot J\, \bra{\psi} H_0 \ket{\psi}.
\ee
\end{prop}
\begin{proof}
First, we partition the lattice $\Lambda$ into disjoint sets $\C_j(r), \, 1\le j\le N(r)$, such that any two distinct points $u,v \in C_j(r)$ satisfy $b_u(r)\cap b_v(r)=\emptyset$. We treat the case $8r > L$ separately, using the trivial decomposition $C_j(r) = \{u_j\}, \, u_j \in \Lambda, \, 1\le j \le N(r) = L^d$.
To accomplish the decomposition for $8r \le L$, we start by dividing $\Lambda$ into an even number of boxes of size $2r'$, in each direction, such that $2r \le 2r' \le 3r$. In particular, we have $L = 2rq+r_0, \, 0\le r_0 < 2r$, so we may choose $2r' = 2r q/(q-b) + r_0/(q-b)$, where $b \equiv q \pmod 2$. Since we assumed $8r \le L$, we know that $2r\cdot 3+r_0 < L \implies q \ge 4$, which implies the bounds on $2r'$. Note that $$L/2r' = q-b \equiv 0 \pmod 2$$ and that $r'$ need not be an integer.

\begin{figure}
\includegraphics[width=8cm]{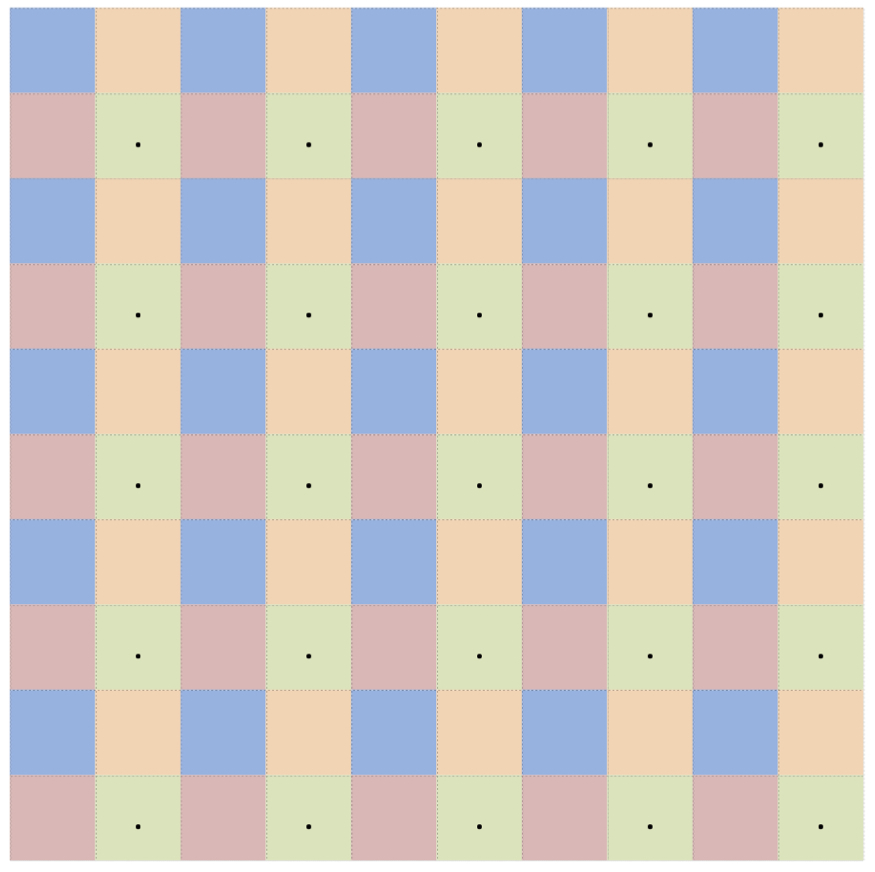}
\begin{center}
\caption{{\small{Partition of the lattice $\Lambda$ into an even number of boxes of size $2r'$ each. The black dots indicate the points $u^a_j$ inside some set $C_j(r)$. The label $a$ corresponds to the label of the box in which the point is found.}}}
\label{fig:partition}
\end{center}
\end{figure}

We build the partition one dimension at a time, as follows: In one dimension, we ``color'' every other interval of length $2r'$ on the line with the same color. So, in one dimension, we have $2$ colors. In $2$ dimensions, we fill the plane with an even number of rows of $2r' \times 2r'$ squares, where every other row looks the same (same ``color'' pattern) and within each row we have an even number of alternating colored boxes (see Fig.~\ref{fig:partition}.) We can continue like this for any dimension $d$ of our lattice, requiring $2^d$ colors in total. Now, because of the way we have constructed our partition, any two points on distinct boxes of the same color will be a distance at least $2r' \ge 2r$ apart, as desired. By choosing exactly one point $u^a_j$ from each box $B_a$ of the same color, we get a family of disjoint sets $\C_j(r)$ that cover $\Lambda$, each satisfying the desired property of internal sparsity. The total number $N(r)$ of these sets is bounded by the number of colors multiplied by the volume of each box. So, we have:
\be\label{N(r):bnd}
N(r) \le 2^d (2r')^d \le 6^{d} r^d,
\ee
and 
\be
\Lambda = \bigcup_{j \ge 1}^{N(r)} \C_j(r), \quad \mbox{with } \C_j(r)=\{u^a_j, \, a\ge 1\} \mbox{ defined above.}
\ee
From this point onward, we set:
\be
b_{j,a}(r) := b_{u^a_j}(r), \quad u^a_j \in \C_j(r), \quad a \ge 1.
\ee
Note that $[P_{b_{j,a}(r)},P_{b_{j,b}(r)}]=0$, since the projectors have disjoint supports for $a \neq b$.
Define an orthogonal decomposition of unity as follows:
\[
R^{\mathcal{Y}}_j(r)=\prod_{a=1}^{|\C_j(r)|}\left(\mathcal{Y}_{a}(1-P_{b_{j,a}(r)})+(1-\mathcal{\mathcal{Y}}_{a})P_{b_{j,a}(r)}\right), \quad \mathcal{Y} \in \{0,1\}^{|\C_j(r)|}.\]
It is easy to check that $\sum_{\mathcal{Y}} R^{\mathcal{Y}}_j(r) = \one$ and $R^{\mathcal{Y}}_j(r) R^{\mathcal{Z}}_j(r) = \delta_{\mathcal{Y,Z}} R^{\mathcal{Y}}_j(r).$

Before continuing, we define a decomposition of the perturbation $W$ as follows:
\bea\label{W:decomposition}
W &=& \sum_{k \ge 1}^M \sum_{j\ge 1}^{N(r_k)} \W_j(r_k), \quad \W_j(r_k) := \sum_{a: u^a_j \in \C_j(r_k)} \W_{u^a_j}(r_k),\\
 \W_{u}(r_k) &:=& \sum_{r > r_{k-1}}^{r_k} W_u(r),  \quad 0 = r_0 < r_1 \le r_2\le \ldots \le r_M = L.
\eea
Since $W(r)$ is locally block-diagonal, we have $$W_{u^a_j}(r) P_{b_{j,a}(r)} = 0 \implies \W_{u^a_j}(r_k) P_{b_{j,a}(r_k)}=0.$$ 
Moreover, recalling that $w(r) := \sup_u \|W_{u}(r)\|/J$, we have:
\be
\|\W_{u^a_j}(r_k)\| \le J \sum_{r > r_{k-1}}^{r_k} w(r).
\ee
Setting $\hat{w}(r_k) := \sum_{r > r_{k-1}}^{r_k} w(r)$, we will show that:
\be
|\bra{\psi}\W_j(r_k)\ket{\psi}| \le J \cdot \hat{w}(r_k)\, \bra{\psi} G_{j}(r_k)\ket{\psi},
\ee 
where,
\be
G_{j}(r_k):=\sum_{n\ge1}n\cdot\sum_{|\mathcal{Y}|}R^{\mathcal{Y}}_j(r_k).
\ee
Define $\W_{j,a}(r_k) := \W_{u^a_j}(r_k)$, for $u^a_j \in \C_j(r_k)$. By definition of $R^j_\mathcal{Y}(r)$ and the local block-diagonality of $\W_j(r_k) = \sum_{a} \W_{j,a}(r_k)$, we have:
\begin{align}
&R^{\mathcal{Y}}_j(r_k)\W_{j}(r_k)R^{\mathcal{Z}}_j(r_k) = \delta_{\mathcal{Y,Z}} \sum_{a: \mathcal{Y}_a =1} R^{\mathcal{Y}}_j(r_k)\W_{j,a}(r_k)R^{\mathcal{Y}}_j(r_k) \implies \\
& \|R^{\mathcal{Y}}_j(r_k)\W_{j}(r_k)R^{\mathcal{Z}}_j(r_k)\| \le J \, \hat{w}(r_k) \, \delta_{\mathcal{Y,Z}}\cdot  |\mathcal{Y}|,\label{W:bound2}
\end{align}
where the equality in the first line follows from the observations: 
\begin{enumerate}
\item $\W_{j,a}(r_k)R^{\mathcal{Y}}_j(r_k) = 0$, for $\mathcal{Y}_a = 0$.
\item If $\mathcal{Y} \neq \mathcal{Z}$, then either $\mathcal{Z}_a = 0$ when $\mathcal{Y}_a=1$ and $\W_{j,a}(r_k)R^{\mathcal{Z}}_j(r_k) = 0$, or $\mathcal{Y}_b = 0$ and $\mathcal{Z}_b=1$, for $b\neq a$, so $R^{\mathcal{Y}}_j(r_k) \W_{j,a}(r_k) R^{\mathcal{Z}}_j(r_k) $ equals:
$$P_{b_{j,b}(r_k)}\left(1-P_{b_{j,b}(r_k)}\right) R^{\mathcal{Y}}_j(r_k) \W_{j,a}(r_k) R^{\mathcal{Z}}_j(r_k) = 0.$$
\end{enumerate}
The bound (\ref{W:bound2}) implies for arbitrary states $\ket{\psi}$: 
\begin{align*}
&|\bra{\psi} \W_{j}(r_k)\ket{\psi}| \le \sum_{\mathcal{Y}} |\bra{\psi} R^{\mathcal{Y}}_j(r_k)\W_{j}(r_k)R^{\mathcal{Y}}_j(r_k) \ket{\psi}|\\
&\le \sum_{\mathcal{Y}} \|R^{\mathcal{Y}}_j(r_k)\W_{j}(r_k)R^{\mathcal{Y}}_j(r)\| \bra{\psi} R^{\mathcal{Y}}_j(r_k) \ket{\psi}
 \le J  \hat{w}(r_k) \sum_{\mathcal{Y}} |\mathcal{Y}| \bra{\psi} R^{\mathcal{Y}}_j(r_k) \ket{\psi} \\
&= J  \hat{w}(r_k) \bra{\psi} G_{j}(r_k)\ket{\psi}.
\end{align*}
Recalling that $H_{b_{j,a}(r)} = \sum_{u: b_u(1) \subset b_{j,a}(r)} Q_u$, define $H_j(r) = \sum_a H_{b_{j,a}(r)}$ and note that: 
\be\label{bnd:color}
\sum_{j=1}^{N(r)} H_j(r) = \sum_{u\in \Lambda} \Big(\sum_{j\ge 1} {\rm Ind}\big(b_u(1) \subset \cup_{a\ge 1} b_{j,a}(r)\big)\Big) Q_u \le C_d \,r^d H_0,
\ee
where $C_d$ is the volume of the unit ball in $\Z^d$ and ${\rm Ind}$ denotes the indicator function. To see this, note that $\sum_{j\ge 1} {\rm Ind}\big(b_u(1) \subset \cup_{a\ge 1} b_{j,a}(r)\big)$ is bounded above by the number of balls of radius $r$ that contain $b_u(1)$, since ${\rm Ind}\big(b_u(1) \subset \cup_{a\ge 1} b_{j,a}(r)\big)$ is non-zero only when $j$ is such that $b_{j,a}(r) := b_{u_j^a}(r)$ includes $b_u(1)$ for some unique $a$, since $b_{j,a}(r) \cap b_{j,b}(r) = \emptyset$, for $a\neq b$. Thus, the indicator sum is bounded by the number of places we can put $b_u(1)$ inside a ball of radius $r$, which is the same as the number of balls of radius $r$ containing $b_u(1)$. The former is upper-bounded by the number of places we can put $b_0(u) = u$ inside a ball of radius $r$, which is exactly equal to $C_d\, r^d$.

Using the assumption that $H_0$ is locally gapped and frustration-free, we have $$H_{b_{j,a}(r)} \ge \gamma_{j,a}(r) (1-P_{b_{j,a}(r)}), \quad \mathcal{Y}_a = 0 \implies H_{b_{j,a}(r)} R^j_{\mathcal{Y}}(r) = 0.$$
Then, we can proceed as follows, for arbitrary states $\ket{\psi}$: 
\begin{align*}
&\bra{\psi} H_j(r) \ket{\psi} = \sum_{\mathcal{Y}} \sum_{a: \mathcal{Y}_a =1} \bra{\psi} R^{\mathcal{Y}}_j(r) H_{b_{j,a}(r)} R^{\mathcal{Y}}_j(r)\ket{\psi}\\
&\ge \sum_{\mathcal{Y}} \sum_{a: \mathcal{Y}_a =1} \gamma_{j,a}(r) \bra{\psi} R^{\mathcal{Y}}_j(r) (1-P_{b_{j,a}(r)}) R^{\mathcal{Y}}_j(r)\ket{\psi}\\
& \le \gamma_{j}(r) \sum_{n \ge 1} n \sum_{|\mathcal{Y}|} \bra{\psi}R^{\mathcal{Y}}_j(r)\ket{\psi}
\implies G_j(r_k) \le H_j(r_k)/ \gamma_j(r_k),
\end{align*} 
where $\gamma_{j}(r) := \min_{a} \gamma_{j,a}(r)$.
We have just shown that: $$|\bra{\psi} \W_j(r_k) \ket{\psi}| \le J \hat{w}(r_k) \bra{\psi} G_j(r_k) \ket{\psi} \le J \hat{w}(r_k)/\gamma_j(r_k)\, \bra{\psi} H_j(r_k)\ket{\psi}.$$ 
This implies that: 
\begin{align}
|\bra{\psi} W \ket{\psi}| \le \sum_{k = 1}^M \sum_{j=1}^{N(r_k)} |\bra{\psi} \W_j(r_k) \ket{\psi}|
\le J \, \sum_{k=1}^M \hat{w}(r_k) \sum_{j=1}^{N(r_k)}  \bra{\psi} H_j(r_k) \ket{\psi}/\gamma_j(r_k).
\end{align}
Setting $\gamma(r) := \min_j \gamma_j(r)$, we have:
\begin{align}
|\bra{\psi} W \ket{\psi}| \le c\, J \bra{\psi} H_0 \ket{\psi},
\end{align}
where $c = C_d \sum_{k=1}^{M} r_k^d\, \hat{w}(r_k)/\gamma(r_k)$, using (\ref{N(r):bnd}) and (\ref{bnd:color}).
\end{proof}

We note a few things about the definition of the structural constant $c$ appearing in Proposition~\ref{prop:relative_bound}: (i) The choice of the $r_k: 0 < r_1 \le \ldots \le r_M = L$ is left unspecified beyond the boundary conditions, so that one may optimize the constant $c$ given the relevant decays $w(r)$ and $\gamma(r)$. For example, if $\gamma(r)$ decays particularly fast for some values of $r$, we may want to avoid these values by summing over them within the rapid decay of $w(r)$. (ii) Assuming polynomial decay of $\gamma(r_k)$ for the selected values of $r_k$, it is clear that for fast enough decay of $w(r)$ (superpolynomial, or even fast enough polynomial decay), $c$ is a constant depending only on the dimension $d$ of the lattice and the decay rate of $\gamma(r_k)$ and $w(r)$.

\section{Final Argument:}\label{sec:theorem}
The argument below follows the bootstrapping argument in ~\cite{bravyi:2010b}, putting together all the pieces to prove stability of frustration-free systems under weak perturbations.

\begin{repthm}{thm:main_result}
Let $H_0$ be a {\it frustration-free} Hamiltonian, with spectral gap $\gamma$, defined in Section \ref{sec:ham}, satisfying the {\it Local-TQO} and {\it Local-Gap} conditions. For a $(J,f_0)$-perturbation $V$, there exist constants $J_0 >0$ and $L_0 \ge 2$, given in Proposition \ref{prop:relative_bound}, such that for $J \le J_0$ and $L \ge L_0$, the spectral gap of $H_0+V$ is bounded from below by $\gamma/2$.\end{repthm}

\begin{proof}
Let $J_0 = \frac{1}{3c}$, where $c$ is defined in Proposition~\ref{prop:relative_bound}.
Assume that for $J \le J_0$, $H_0+sV$ has gap $\gamma_s \ge \gamma/2$, for $0\le s \le s_0 < 1$ and $\gamma_s < \gamma/2$, for $s_0 < s \le s_1\le 1$. Then, we have proven in the previous sections that for $0\le s \le s_0$: 
$$U^\dagger(s) (H_0+sV) U(s) = H_0 + W_s + \Delta_s + E_s \cdot \one,$$ where $U(s)$ generates the {\it spectral flow} $U(s) P_0(0) U^\dagger(s) = P_0(s)$, and $|\bra{\psi} W_s \ket{\psi}| \le c \cdot J\,\bra{\psi} H_0 \ket{\psi}$. Moreover, $E_s \cdot \one$ is an overall energy shift and the norm of $\Delta_s$ decays rapidly in $L$. We note at this point that the \emph{Local-TQO} cut-off $L^*$ plays a crucial role in the decay estimates of $\|\Delta_s\|$, for example. In particular, $L^*$ should scale as $L^\alpha$, for some $\alpha \in (0,1]$, in order for estimates like the one in Lemma \ref{lemma:D_u} to decay rapidly in $L$, especially since the estimate for $\|\Delta_s\|$ is given by $L^d \, \sup_{u\in\Lambda} \|\Delta_u\| \sim L^d \, \max\{\Delta_0(L^*),f_0(L^*)\}$.

We will show that, for large enough $L$, the gap of $H_0+W_{s_0}+\Delta_{s_0}$, which is $\gamma_{s_0}$, is strictly greater than $\gamma/2$, contradicting the assumption that $\gamma_s < \gamma/2$, for $s_0 < s \le s_1$, implying that $s_0 =1$. 

Now, we note that $H_0 + W_s+ \Delta_s = U^\dagger(s) (H_0+sV-E_s \cdot \one) U(s)$ has eigenstates in the subspace $P_0(0)$, with energy within $\|\Delta_s\|$ from $0$. Those states correspond to the ``groundstate'' subspace $P_0(s)$ of the unrotated version of the Hamiltonian, $H_0+sV$. In general, eigenstates of $H_0+sV$ in the $P_0(s)$ subspace will have different energies, but only up to vanishingly small error in the system size, $L$.
To see this, note that for $\ket{\Psi_0(s)} = P_0(s) \ket{\Psi_0(s)}$ an eigenstate of $H_0+sV$ with energy $E_0(s)$, setting $\ket{\Psi_0} := U^\dagger(s) \ket{\Psi_0(s)}$, we have $P_0(0) \ket{\Psi_0} = U^\dagger(s) P_0(s) \ket{\Psi_0(s)} = \ket{\Psi_0}$ and $$(H_0+W_s+\Delta_s)\ket{\Psi_0} = U^\dagger(s) (H_0 + sV - E_s\cdot \one) U(s) \ket{\Psi_0} = (E_0(s)-E_s) \ket{\Psi_0}.$$ 
Moreover, since $|\bra{\Psi_0} W_s \ket{\Psi_0}| \le cJ \bra{\Psi_0} H_0 \ket{\Psi_0} = 0$, the energy $E_0(s)$, satisfies $$(H_0+W_s+\Delta_s) \ket{\Psi_0} = \Delta_s \ket{\Psi_0} = E_0(s) \ket{\Psi_0}.$$ 
Hence, $|E_0(s)|\le \|\Delta_s\| \le \gamma/13,$ for large enough $L$. 

Moreover, any state $\ket{\psi_1}$ orthogonal to the $P_0(0)$ subspace has energy at least $(1-cJ) \gamma - \|\Delta_{s}\|$, which follows from the linear relative bound on $W_s$ and the gap of $H_0$:
\begin{eqnarray*}
\bra{\psi_1} H_0+W_s+\Delta_s \ket{\psi_1} &\ge& \bra{\psi_1}H_0 \ket{\psi_1} - |\bra{\psi_1}W_s\ket{\psi_1}| - |\bra{\psi_1} \Delta_s \ket{\psi_1}|\\
&\ge& (1-cJ) \bra{\psi_1}H_0 \ket{\psi_1} - \|\Delta_s\| \ge (1-cJ) \gamma - \|\Delta_s\|.
\end{eqnarray*}
Now, taking the linear size $L$ large enough (such that $\|\Delta_s\| \le \gamma/13$) and recalling that $cJ \le 1/3$, we see that the gap is bounded below by: $$\gamma_s \ge (1-cJ) \gamma - 2\|\Delta_{s}\| \ge (20/39)\gamma > \gamma/2,$$ for all $s \in [0,s_0]$. But, by definition of $s_0$, we know that for any $0 < \epsilon \le s_1-s_0$, we have $\gamma(s_0+\epsilon) < \gamma/2$, a contradiction of the continuity of $\gamma_s$ for finite $L$. Hence, $s_0=1$ and $\gamma(1) \ge \gamma/2$, for all $J \le J_0$ and $L$ large enough to satisfy $\|\Delta_s\| \le \gamma/13$.

Finally, this implies that $H_0 + sV$ has gap $\gamma_s\ge \gamma/2$, for all $s \in [0,1]$ for the parameters of $J_0$ and $L$ given above and determined implicitly by the decay $f_0(r)$ of $V$, the \emph{Local-TQO} decay $\Delta_0(r)$, the local gap $\gamma(r)$ and the linear size $L$ of the system.
\end{proof}
{\bf Note on the stability of general gapped eigenspaces:}
We note that the above argument goes through for all gapped eigenspaces of the unperturbed Hamiltonian $H_0$, since the {\it quasi-adiabatic evolution} $U_s$ and the relative bound of Proposition~\ref{prop:relative_bound} work for any gapped eigenspace, though the constants $J_0$ and $L_0$ may be different. In particular, $J_0$ will be inversely proportional to the energy of the eigenspace under consideration (see the main theorem in \cite{bravyi:2010a} for comparison with the commuting case) and $L_0$ will depend on the spectral gap of the particular eigenspace we consider.

\section{Discussion}\label{sec:discussion}
We have shown that topological phases corresponding to the groundstate sector of gapped, frustration-free Hamiltonians are stable under quasi-local perturbations. In particular, the splitting in the energy of the groundstate subspace is bounded by a rapidly-decaying function of the system size and the gap to the high-energy sectors remains open for perturbation strengths independent of the system size. Our result also implies the stability of symmetry-protected sectors under perturbations that obey the symmetry of that sector. To see this, note that the \emph{Local-TQO} condition is only required for operators $O_A$ obeying the symmetries under consideration. Stability follows using the methods already found here, as is already discussed in~\cite{bravyi:2010b}.

Moreover, it can be shown \cite{osborne:2010} that parent Hamiltonians of Matrix Product States (MPS) satisfy the \emph{Local-TQO} condition with $\Delta_0$ decaying exponentially in these cases. Combined with the results of Nachtergaele \cite{nachtergaele:1996} and Spitzer \etal \cite{spitzer:2002} on the spectral gap of one-dimensional, frustration-free Hamiltonians, our result shows that parent Hamiltonians of MPS have stable low-energy spectrum against arbitrary, sufficiently weak, quasi-local perturbations.

Finally, we note that Lemma~\ref{lemma:global} applied to any local Hamiltonian $H'_0=\sum_u Q_u$ with spectral gap $\gamma' > 0$, transforms the initial local terms $Q_u$ into quasi-local terms $\Fz(Q_u)$ which commute with the groundstate subspace $P_0$. Subtracting the groundstate energy $E_0(Q_u) = \braket{\psi_0}{Q_u\psi_0}$ from $\Fz(Q_u)$ and taking the absolute value of each term, we get the new frustration-free Hamiltonian $H_0 = \sum_u |\Fz(Q_u)-E_0(Q_u)|$ with groundstate subspace $P_0$ and spectral gap $\gamma \ge \gamma'$ (since $H_0 \ge H'_0-E_0(H'_0)$). The new Hamiltonian has quasi-local interactions, which can be seen from the operator-norm inequality $\| |X+Y|-|X|\| \le \left(a + b \ln \frac{\|X\|+\|Y\|}{\|Y\|}\right) \|Y\|$, for constants $a,b >0$ \cite{kato:1973}, where $X$ corresponds to a localized version of $\Fz(Q_u)$ and $Y$ denotes the rapidly-decaying (in norm) remainder of the interaction. This observation suggests a path towards characterizing the stability of general gapped Hamiltonians satisfying {\it Local-TQO}: After transforming the initial gapped Hamiltonian $H'_0$ into the frustration-free, gapped Hamiltonian $H_0$, one can truncate the quasi-local terms of $H_0$ into strictly local terms at the cost of a rapidly decaying error in the localization length. At this point, one has a nearly frustration-free, gapped Hamiltonian and a more general {\it Local-TQO} condition (where $P_B(\epsilon)$ is used instead of $P_B=P_B(0)$ in Definition~\ref{def:local_tqo}) may be sufficient for a proof of stability along the lines of the current work.

We conclude with a list of open questions related to the current work and future extensions:
\begin{enumerate}
\item First, can we extend the one-dimensional methods, introduced by Nachtergaele in \cite{nachtergaele:1996}, to prove rigorous spectral gaps for frustration-free systems in two dimensions?
\item Can we use the decay rate $\Delta_0$ in \emph{Local-TQO} to classify different phases of matter? In other words, is there a canonical parent Hamiltonian construction, given a groundstate subspace, such that \emph{Local-TQO} is satisfied with an optimal decay $\Delta_0$? For example, see~\cite[Sec. 2.4]{bravyi:2010a}, where two different Hamiltonians for the \emph{toric code} give very different decays $\Delta_0$; one being constant and the other being a step function.
\item Under what conditions (maybe translation invariance, and/or scale invariance) is \emph{Local-TQO} satisfied, if the global groundstate space is topologically ordered? This is related to the previous question.
\item Can one rigorously prove that all gapped, frustration-free systems with non-vanishing, local interactions, satisfy the \emph{Local-Gap} condition, using the \emph{detectability lemma}, or some other method? 
\item What features of parent Hamiltonians of PEPS can one use to prove that the \emph{Local-TQO} and \emph{Local-Gap} conditions are satisfied?
\item Can one use the present result to prove the Haldane conjecture, by identifying a gapped, frustration-free Hamiltonian which satisfies the \emph{Local-TQO} and \emph{Local-Gap} conditions, sufficiently close to the Haldane phase \cite{kennedy:1992a}?
\item Can we identify two-dimensional, frustration-free systems satisfying \emph{Local-TQO} that exhibit Topological Order at non-zero temperatures? What about three-dimensional frustration-free systems?
\item Finally, what is the form of \emph{Local-TQO} that is truly necessary (and sufficient) for stability of the spectral gap against local perturbations, for frustrated, gapped Hamiltonians? Is the \emph{Local-TQO} condition always true for generic gapped Hamiltonians?
\end{enumerate}

\subsection*{Acknowledgments:}
S. M. would like to thank Sergey Bravyi, Matt Hastings and Tobias Osborne for encouraging the pursuit of this result, Frank Verstraete for pointing out the connection of \emph{Local-TQO} with the area-law for the entanglement entropy, Steve Flammia and Norbert Schuch for working out the details of that connection with the author, Robert Koenig and Bruno Nachtergaele for discussions on some of the technical details and Alexei Kitaev for several useful discussions on the extension of this result to general, gapped Hamiltonians. S.M. received support from NSF \#DMS-0757581 and \#PHY-0803371, and DOE Contract \#DE-AC52-06NA25396. S. M. and J. P. are grateful for the warm hospitality of Los Alamos National Lab, where part of this work was completed during the summer of $2010$, as well as to the referees whose thoughtful suggestions resulted in substantial improvement for both the clarity and the scope of this paper.

\appendix
\section{Proof of Corollary \ref{cor:local_tqo}}\label{append:local_tqo}
From \emph{Local-TQO} and the triangle inequality, we have:
\begin{align}
|\left\Vert O_{A}P_{A(\ell)}\right\Vert ^{2} - c_\ell(|O_A|^2)| &\le \|O_{A}\|^{2}\,\Delta_{0}(\ell)\label{local_tqo},\\
|\left\Vert O_{A}P_0 \right\Vert ^{2} - c_{L-r}(|O_A|^2)| &\le \|O_{A}\|^{2}\,\Delta_{0}(L-r)\label{global_tqo}, \\
| c_{L-r}(|O_A|^2) -  c_{\ell}(|O_A|^2) | &\le  \|O_{A}\|^{2}\, (\Delta_{0}(\ell) + \Delta_{0}(L-r) ).\label{trace_tqo}
\end{align}
The last inequality follows from a triangle inequality on (\ref{global_tqo}) and the bound:
\begin{align*}
\|P_0 |O_{A}|^2 P_0 - c_{\ell}\left(|O_{A}|^2\right) P_0\| = \|P_0 P_{A(\ell)}\left[|O_{A}|^2-c_{\ell}\left(|O_{A}|^2\right)\right]P_{A(\ell)} P_0\|\le
 \|P_{A(\ell)}|O_{A}|^2P_{A(\ell)} -c_{\ell}\left(|O_{A}|^2\right) P_{A(\ell)}\|.
\end{align*}
Recalling that $\Delta_{0}(\ell)$ is decaying, we have:
\begin{align*}
\left|\left\Vert O_{A}P_{A(\ell)}\right\Vert -\left\Vert O_{A}P_0\right\Vert \right|^{2} \le \left|\left\Vert O_{A}P_{A(\ell)}\right\Vert ^{2}-\left\Vert O_{A}P_0\right\Vert ^{2}\right|
\le 2 \|O_{A}\|^{2}\,(\Delta_{0}(\ell)+\Delta_{0}(L-r)) \le 4 \|O_{A}\|^{2}\,\Delta_{0}(\ell).
\end{align*}
The second bound is equivalent to showing that $\|(1- P_A) P_{A(\ell)}\|^2 \le 3 \Delta_{0}(\ell)$. If we are given two projections $P$ and $Q$ satisfying: $\|PQP-P\| \le \delta$, then we have $\|P(1-Q)P\| \le \delta \implies \|(1-Q)P\|^2 \le \delta$. In our case, we have from \emph{Local-TQO}: $\|P_{A(\ell)} P_A P_{A(\ell)} - P_{A(\ell)}\| \le 2\Delta_{0}(\ell)+\Delta_{0}(L-r)$, where we used Corollary~\ref{cor:local_energy} to get $c_{L-r}(P_A) = \Tr(P_AP_0)/\Tr P_0 = 1$ and combined it with (\ref{trace_tqo}) and a triangle inequality.

\section{Proof of Corollary \ref{cor:local_same}}\label{append:local_same}
First, we prove that \emph{Local-TQO} implies the above condition.
Set $\rho_i = \pure{\psi_i(A(\ell))}, \, i=1,2$. Then, $\rho_i(A) = \Tr_{A(\ell)\setminus A} \rho_i$ and is supported only on $A$. From \emph{Local-TQO}, for any operator $O_A$ supported on $A$:
\begin{eqnarray*}
|\Tr(\rho_1-\rho_2) O_A |
= |\bra{\psi_1(A(\ell))} O_A  -c_\ell(O_A) \ket{\psi_1(A(\ell))} - \bra{\psi_2(A(\ell))} O_A -c_\ell(O_A) \ket{\psi_2(A(\ell))}|
\le 2 \|O_A\| \Delta_{0}(\ell).
\end{eqnarray*}
Now, recalling that $\|\rho_1(A) - \rho_2(A)\|_1 = \sup_{O_A: \|O_A\|=1} |\Tr(\rho_1-\rho_2) O_A|$, we get the desired bound.
To prove the other direction, we note that it is sufficient to show that for $\ket{\Psi_0(A(\ell))}$ satisfying $P_{A(\ell)}\ket{\Psi_0(A(\ell))}=\ket{\Psi_0(A(\ell))}$, we have:
\be
|\bra{\Psi_0(A(\ell))} O_A \ket{\Psi_0(A(\ell))} - c_{\ell}(O_A)| \le 2\|O_A\| \Delta_0(\ell).
\ee
Setting $N = \Tr P_{A(\ell)}-1$ and $\rho_k(A) = \Tr_{A(\ell)\setminus A} \pure{\Psi_k(A(\ell))}, \quad 0 \le k \le N,$ where $ \{\pure{\Psi_k(A(\ell))}\}_{k = 0}^N$ is an orthogonal decomposition for $P_{A(\ell)}$, we have for $\|O_A\|=1$:
\begin{align*}
&|\bra{\Psi_0(A(\ell))} O_A \ket{\Psi_0(A(\ell))} - c_{\ell}(O_A)| = |\Tr O_A \rho_0(A) - c_{\ell}(O_A)|=\\ 
& \frac{1}{N+1} |\sum_{k=0}^N \Tr O_A (\rho_0(A) - \rho_k(A))| \le \frac{1}{N+1} \sum_{k=1}^N \|\rho_0(A)-\rho_k(A)\|_1 \le \frac{N}{N+1} \Delta_0(\ell).
\end{align*}
Hence, up to constant pre-factors, (\ref{local_tqo:2}) and \emph{Local-TQO} are equivalent.

\section{Lieb-Robinson Estimates for Lemma~\ref{lemma:local}}\label{lr-estimates2}
We follow the proof of Thm. $5.2$ in~\cite{BMNS:2011} to bound $\|\alpha_s^{b_u(r')}(O_u(r)) - \alpha_s^{b_u(r'-1)}(O_u(r))\|$ in Lemma~\ref{lemma:local}. To use the notation of ~\cite{BMNS:2011}, we set $\Lambda_n = b_{u}(n+r)$, for $n \ge 0$, so that $O_u(r)$ is supported on the set $\Lambda_{0}$. We have:
\begin{align}
&\|\alpha_s^{\Lambda_{r'}}(O_u(r)) - \alpha_s^{\Lambda_{r'-1}}(O_u(r))\| \le  \int_0^s \Vert [D_{\Lambda_{r'}}(t)-D_{\Lambda_{r'-1}}(t),\alpha_t^{\Lambda_{r'-1}}(O_u(r))]\|\, dt \\
&\le \left(2\, C_F \, \| \Psi \|_F \, \int_0^s g(t) \, dt \, \right) \|O_u(r) \| \,  \sum_{z\in\Lambda_{r'}\setminus\Lambda_{r'-1}} \sum_{x \in \Lambda_{0}} F_\Psi(d(x,z))\label{good_estimate1} \\
&+2\, \left(2\, C_F \, \| \Psi \|_F \, \int_0^s g(t) \, dt \right) \|O_u(r) \| \, \sum_{z \in\Lambda_{r'-1}\setminus \Lambda_{\lfloor (r'-1)/3\rfloor}} \sum_{x \in \Lambda_{0}} F_\Psi(d(x,z))\label{good_estimate2}\\
&+4 s\, \|\Psi\|_F \, \|O_u(r)\| \, \sum_{z\in\Lambda_{r'-1}\setminus \Lambda_{\lfloor 2(r'-1)/3\rfloor}} \sum_{x \in \Lambda_{\lfloor (r'-1)/3\rfloor}} F_\Psi(d(x,z))\label{good_estimate3}\\
&+16 s \, \|\partial \Phi \|_a F_a(0) \sum_{n \geq 1} \sqrt{G_1(n-1)}\, \|O_u(r)\| \, | \Lambda_{\lfloor 2(r'-1)/3\rfloor} | \sqrt{K(\lfloor(r'-1)/3\rfloor)},\label{bad_estimate}
\end{align}
where the following functions are defined:
\begin{equation}\label{eq:definitionFpsi}
F_{\Psi}(r) = \tilde{u}_{\mu} \left( \frac{\gamma'}{8 v_a} r \right) \,F\left(\frac{\gamma'}{8 v_a} r\right),
\end{equation} 
with $F(r) = \left(1+ r  \right)^{-(d+1)}$ and,
\bea
\tilde{u}_{\mu}(x) &=& \left\{ \begin{array}{cc} u_{\mu}(e^2) & \mbox{for } 0 \leq r \leq e^2, \\ u_{\mu}(x) & \mbox{otherwise}. \end{array} \right\},\\ 
u_{\mu}(x)&=&\exp \left\{-\mu\frac{x}{\ln^2 x}\right\}\quad \mbox{ for } \mu > 0, \quad x > 1.
\eea
Moreover,
\begin{equation} \label{eq:defg}
G_1(n) \le 4 G_2\left( \frac{\gamma' n}{2 v_a} \right) + \frac{ C_a \| F \|}{av_a} e^{-an/2} 
\end{equation}
and
\begin{equation}
K(x) \le 4 G_2\left(\frac{\gamma' |x|}{2 v_a} \right) + \frac{C_a C_{F_a} \| \Phi \|_a \| F \|}{a^2 v_a^2} e^{-a x/2},
\end{equation}
where $G_2(\zeta)$ is defined for $\zeta\geq 0$ by
$$
G_2(\zeta) = \frac{1}{\gamma'} \cdot \begin{cases}
7,354 & 0\leq \zeta \leq 36,058 \\
130\ep{2}\zeta^{10} u_{2/7}(\zeta) & \zeta > 36,058.
\end{cases}\,.
$$
Finally, the remaining constants (e.g. $C_F = 2^{d+1} \sum_{r\in \Z^d} F(r)$, $\| \Psi \|_F$, $\|\partial \Phi \|_a$, $F_a(0)$, $a$, $v_a$, etc.) are defined in~\cite[Sec. 4]{BMNS:2011}, where one also sees that:
\begin{equation}
2 C_{F} \| \Psi \|_F \, \int_0^s g(t)\, dt  \le e^{2s \| \Psi \|_{F} C_{F}} -1.
\end{equation}

We note that the decays involved (i.e. the functions $K(x)$ and $F_\Psi(r)$) are most likely not optimal, but we consider it an important exercise to give a rough estimate of the order of magnitude and the dependence to the spectral gap $\gamma'$, appearing in Definition \ref{defn:path}, of the constants involved in the decay of the perturbations we are studying. In any case, the relevant decay for our computation in Lemma \ref{lemma:local} is dominated by tail estimates of the function $u_{\mu}(r')$, which defines a sub-exponential decay with $\mu \sim \gamma'$. Here, we have sketched how one may derive rigorous estimates, using the tightest Lieb-Robinson bounds available at this point in time.

\end{document}